\documentclass{amc}
\usepackage{amsmath}
\usepackage{lmodern}
  \usepackage{paralist}
  \usepackage{graphics} 
  \usepackage{epsfig} 

  \usepackage{hyperref} 

\def\currentvolume{X}
 \def\currentissue{X}
  \def\currentyear{200X}
   
    \def\ppages{X--XX}

\newtheorem{theorem}{Theorem}[section]
\newtheorem{corollary}{Corollary}

\newtheorem{lemma}[theorem]{Lemma}
\newtheorem{proposition}{Proposition}

\newtheorem*{problem}{Problem}
\theoremstyle{definition}
\newtheorem{definition}[theorem]{Definition}
\newtheorem{remark}{Remark}
\newtheorem{example}{Example}

\usepackage{amssymb}
\newcommand{\unitsof}[1]{{#1^{\times}}}
\newcommand{\field}[1]{\mathbb{F}_{#1}}

\newcommand{\relvar}[2]{\buildrel {#2} \over {#1}}
\newcommand{\eqvar}[1]{\relvar{=}{\mathrm{#1}}}
\newcommand{\eqdef}{\eqvar{\scriptscriptstyle\triangle}}
\newcommand{\grmat}{\mathfrak{G}} 
\newcommand{\seq}[1]{\mathbf{#1}} 
\newcommand{\rseq}[1]{\tilde{\mathbf{#1}}} 
\newcommand{\mat}[1]{\mathbf{#1}} 
\newcommand{\rmat}[1]{\tilde{\mathbf{#1}}} 
\newcommand{\matset}[1]{\mathcal{#1}} 
\newcommand{\rset}[1]{\tilde{#1}}
\DeclareMathOperator{\supp}{supp}
\DeclareMathOperator{\rank}{rank}
\newcommand{\rankd}{\mathrm{d}}
\newcommand{\matfield}[1]{\matset{F}_{#1}}
\newcommand{\nzmatfield}[1]{\unitsof{\matset{F}_{#1}}}
\newcommand{\gbinom}[3][q]{\genfrac{[}{]}{0pt}{}{#2}{#3}_{#1}} 
\newcommand{\uorder}{\mathrm{O}} 
\DeclareMathOperator{\AG}{AG}
\DeclareMathOperator{\PG}{PG}
\DeclareMathOperator{\Hom}{Hom}
\newcommand{\minsize}{\mathrm{\mu}}
\newcommand{\lhom}{\mathrm{w}_{\mathrm{\ell}}}
\newcommand{\rhom}{\mathrm{w}_{\mathrm{r}}}
\newcommand{\nlhom}{\overline{\mathrm{w}}_{\mathrm{\ell}}}
\newcommand{\nrhom}{\overline{\mathrm{w}}_{\mathrm{r}}}
\newcommand{\LAG}{\AG_{\mathrm{\ell}}}
\newcommand{\RAG}{\AG_{\mathrm{r}}}

\title[Good Random Matrices over Finite Fields]
      {Good Random Matrices over Finite Fields}

\author[Shengtian Yang and Thomas Honold]{}

\subjclass{94B05, 15B52, 60B20, 15B33, 51E21}
\keywords{Linear code, maximum-rank-distance code, MRD code, intersecting code, random matrix, homogeneous weight, dense set of matrices, affine geometry of matrices, blocking set}

\begin{document}
\maketitle

\centerline{\scshape Shengtian Yang}
\medskip
{\footnotesize
 \centerline{Zhengyuan Xiaoqu 10-2-101, Fengtan Road}
   \centerline{Hangzhou 310011, China}
} 

\medskip

\centerline{\scshape Thomas Honold}
\medskip
{\footnotesize
 \centerline{Department of Information Science and Electronic Engineering}
 \centerline{Zhejiang University}
   \centerline{38 Zheda Road, Hangzhou 310027, China}
} %

\bigskip

 \centerline{(Communicated by Ivan Landjev)}

\begin{abstract}
  The random matrix uniformly distributed over the set of all
  $m$-by-$n$ matrices over a finite field plays an important role in
  many branches of information theory. In this paper a generalization
  of this random matrix, called $k$-good random matrices, is
  studied. It is shown that a $k$-good random $m$-by-$n$ matrix with a
  distribution of minimum support size is uniformly distributed over a
  maximum-rank-distance (MRD) code of minimum rank distance
  $\min\{m,n\}-k+1$, and vice versa. Further examples of $k$-good
  random matrices are derived from homogeneous weights on matrix
  modules. Several applications of $k$-good random matrices are given,
  establishing links with some well-known combinatorial problems.
  Finally, the related combinatorial concept of a $k$-dense set of
  $m$-by-$n$ matrices is studied, identifying such sets as blocking
  sets with respect to $(m-k)$-dimensional flats in a certain
  $m$-by-$n$ matrix geometry and determining their minimum size in
  special cases.
\end{abstract}

\section{Introduction}

Let $\field{q}$ be the finite field of order $q$ and $\rmat{G}$ the
random matrix uniformly distributed over the set $\field{q}^{m \times
  n}$ of all $m\times n$ matrices over $\field{q}$. It is well known
that the linear code ensemble $\{\seq{u}\rmat{G}: \seq{u} \in
\field{q}^m\}$ generated by $\rmat{G}$ as a generator matrix
(nontrivial for $m<n$) is good in the sense of the asymptotic
Gilbert-Varshamov (GV) bound \cite{Barg200209}. In fact, the ensemble
$\{\seq{v} \in \field{q}^n: \rmat{G} \seq{v}^T = \seq{0}\}$ generated
by $\rmat{G}$ as a parity-check matrix (nontrivial for $m<n$) is also
good \cite{Gallager196300}. The application of $\rmat{G}$ is not
confined to channel coding. It also turns out to be good for
Slepian-Wolf coding \cite{Csiszar198207}, lossless joint
source-channel coding (lossless JSCC) \cite{Yang200904}, etc.

The success of $\rmat{G}$ in information theory exclusively depends on
the fundamental property
\begin{equation}\label{eq:StrongSCCGood}
  P\{\seq{u}\rmat{G} = \seq{v}\} = q^{-n} \qquad \forall \seq{u} \in
  \field{q}^m \setminus \{\seq{0}\}, \forall\seq{v} \in \field{q}^n.
\end{equation}
Actually, any random matrix satisfying \eqref{eq:StrongSCCGood} has
the same performance as $\rmat{G}$ for channel coding, lossless JSCC,
etc. We call such random matrices good random matrices.

\begin{definition}\label{df:GoodRandomMatrix}
  A \emph{good random matrix} is a random matrix satisfying
  \eqref{eq:StrongSCCGood}. The collection of all good
  random $m\times n$ matrices over $\field{q}$ is denoted by
  $\grmat(m,n,\field{q})$.
\end{definition}

In what follows it is usually safe to identify random matrices, i.e.\
measurable maps $\Omega \to \field{q}^{m\times n}$  defined on
some abstract probability space $(\Omega, \mathcal{A}, P)$
and relative to the full measure algebra on $\field{q}^{m\times n}$,
with their associated probability distributions. In this model
$\grmat(m,n,\field{q})$ is just the set of functions
$f\colon\field{q}^{m\times n}\to\mathbb{R}$ satisfying $f(\mat{A})\geq 0$
for all $\mat{A}\in\field{q}^{m\times n}$ and
$\sum_{\mat{A}:\seq{u}\mat{A}=\seq{v}}f(\mat{A})=q^{-n}$ for all
$\seq{u}\in\field{q}^m\setminus\{\seq{0}\}$, $\seq{v}\in\field{q}^n$.

In some sense, the random matrix $\rmat{G}$ is trivial, since its
distribution has support $\field{q}^{m \times n}$, the set of all
$m\times n$ matrices. Naturally, one expects to find other good random
matrices with distributions of smaller support size. In
\cite{Yang200909}, it was shown that a random $m\times n$ matrix
uniformly distributed over a Gabidulin code, a special kind of
maximum-rank-distance (MRD) code \cite{Delsarte197800,
  Gabidulin198501, Roth199103}, is a good random matrix with a
distribution of support size as small as $q^{\max\{m,n\}}$. However,
this is not the end of the story. Further questions arise. What is the
minimum achievable support size (of the distribution) of a good random
matrix and what is the relation between good random matrices and MRD
codes?

In this paper, we answer the above questions. At first, we prove two
fundamental facts---a random matrix is good if and only if its
transpose is good, and a random matrix uniformly distributed over
an MRD code is good. Using these facts, we show that the minimum
support size of a good random $m\times n$ matrix is $q^{\max\{m,n\}}$
and that a random $m\times n$ matrix with support size
$q^{\max\{m,n\}}$ is good if and only if it is uniformly distributed
over an $(m,n,1)$ MRD code. Based on this result and knowledge about
MRD codes, we also explain how to construct good random matrices with
minimum support size as well as general good random matrices with
larger support size.

Furthermore, we extend the above results to families of special good
random matrices called $k$-good random matrices ($k=1,2,3,\dots$). An
ordinary good random matrix is a $1$-good random matrix in this sense,
and a $(k+1)$-good random matrix is necessarily a $k$-good random
matrix. It is shown that there is a one-to-one correspondence (up to
probability distribution) between $k$-good random $m\times n$ matrices
with minimum support size and $(m,n,k)$ MRD codes.  

Additional examples of $k$-good random matrices are provided by
so-called homogeneous weights on $\field{q}^{m\times n}$, which are
defined in terms of the action of the ring of $m\times m$ (resp.\
$n\times n$) matrices over $\field{q}$ on $\field{q}^{m\times
  n}$. Originally these weight functions had been introduced as a
suitable generalization of the Hamming weight on the prime fields
$\mathbb{Z}_p$ to rings of integers $\mathbb{Z}_n$. Later they were
generalized to other finite rings and modules, including the case
under consideration.

Next some examples are
given to show applications of $k$-good random matrices and their close
relation with some well-studied problems in coding theory and
combinatorics.

Finally we investigate the related (but weaker) combinatorial concept
of $k$-dense sets of $m\times n$ matrices over $\field{q}$. When
viewed as sets of linear transformations, these $k$-dense sets induce
by restriction to any $k$-dimensional subspace of $\field{q}^m$ the
full set of linear transformations from $\field{q}^k$ to
$\field{q}^n$.

We show that $k$-dense sets of $m\times n$ matrices over $\field{q}$
are equivalent to blocking sets with respect to $(m-k)$-dimensional
flats (i.e.\ point sets having nonempty intersection with every
$(m-k)$-dimensional flat) in a certain matrix geometry, which arises
from the module structure of $\field{q}^{m\times n}$ relative to the
right action by the ring of $n\times n$ matrices over $\field{q}$.
Among all $k$-dense subsets $\matset{A}\subseteq\field{q}^{m\times n}$ 
those meeting every flat in the same
number of points are characterized by the fact that the random $m\times n$
matrix $\rmat{A}$ uniformly distributed over $\matset{A}$ is $k$-good.

The study of blocking sets in classical projective or affine
geometries over finite fields is central to finite geometry and has
important applications (to the packing problem of coding theory, for
example). Quite naturally the focus lies on minimal blocking sets and
their properties.  We close the section on $k$-dense sets of matrices
by determining the minimum size of $k$-dense sets (blocking sets with
respect to $(m-k)$-dimensional flats) in certain cases. The example of
binary $3\times 2$ matrices - small enough to be solved by hand
computation, but already requiring a considerable amount of geometric
reasoning - is worked out in detail.

The paper is organized as follows. In
Section~\ref{sec:GoodRandomMatrix}, we study the minimum support size
of good random matrices and their relation to MRD codes. In
Section~\ref{sec:KGoodRandomMatrix}, we introduce and study $k$-good
random matrices. Homogeneous weights and their use for the
construction of $k$-good random matrices are discussed in
Section~\ref{sec:whom}. In Section~\ref{sec:applic}, we discuss the
application of $k$-good random matrices to some well-known
combinatorial problems.
Section~\ref{sec:AlgebraicAspects} contains the material on
$k$-dense sets of matrices, and Section~\ref{sec:Conclusion}
concludes our paper with some suggestions for further work.

In the sequel, if not explicitly indicated otherwise, all matrices and
vectors are assumed to have entries in the finite field $\field{q}$. Vectors
are denoted by boldface lowercase letters such as $\seq{u}$,
$\seq{v}$, $\seq{w}$, which are regarded as row vectors. Matrices are
denoted by boldface capital letters such as $\mat{A}$, $\mat{B}$,
$\mat{C}$. By a tilde we mean that a matrix (resp. vector) such as
$\rmat{A}$ (resp. $\rseq{v}$) is random, i.e.\ subject to some
probability distribution $\mat{A}\mapsto P(\rmat{A}=\mat{A})$ on
$\field{q}^{m\times n}$ (in the case of matrices).\footnote{The
exact nature of the random variables $\rmat{A}$, $\rmat{v}$ does not
matter in our work, so that they can be safely identified with their
distributions (respectively, joint distributions in the case of
several random variables).} Sets of matrices are denoted
by script capital letters such as $\matset{A}$, $\matset{B}$,
$\matset{C}$. For the set $\{\mat{A}\mat{B}: \mat{A} \in \matset{A}\}$
(resp. $\{\mat{A}\mat{B}: \mat{B} \in \matset{B}\}$) of products of
matrices, we use the short-hand $\matset{A}\mat{B}$
(resp. $\mat{A}\matset{B}$). The Gaussian binomial
coefficient $\gbinom{n}{m}$ is defined by
\[
\gbinom{n}{m} \eqdef \left\{\begin{array}{ll}
\displaystyle\frac{\prod_{i=n-m+1}^n (q^i-1)}{\prod_{i=1}^m (q^i-1)}
&\text{for $0\le m \le n$}, \\
0 &\text{otherwise}.\footnotemark
\end{array}\right.
\]
\footnotetext{Here we tacitly assume that an empty product is $1$.}%
Finally, we will say that $\mat{A}\in\field{q}^{m\times n}$ is
of \emph{full rank} if $\rank(\mat{A})=\min\{m,n\}$.

\section{Good Random Matrices and MRD Codes}
\label{sec:GoodRandomMatrix}

For any random matrix $\rmat{A}$, we define the \emph{support} of $\rmat{A}$ by
\begin{equation}
\supp(\rmat{A}) \eqdef \{\mat{A} \in \field{q}^{m \times n}: P\{\rmat{A} = \mat{A}\} > 0\}.
\end{equation}
It is clear that the support size of a good random $m\times n$ matrix
is greater than or equal to $q^{n}$. On the other hand, in
\cite{Yang200909}, we constructed a good random matrix with support
size $q^{\max\{m,n\}}$, which gives an upper bound on the minimum
support size of good random matrices. Combining these two facts then
gives
\begin{equation}\label{eq:MinimumSupportSizeApproximation}
q^n \le \min_{\rmat{A} \in \grmat(m,n,\field{q})} |\supp(\rmat{A})| \le q^{\max\{m,n\}}.
\end{equation}
When $m \le n$, the two bounds coincide and hence give the exact
minimum support size. When $m > n$, however, it is not clear what
the minimum support size is. To answer this question, we need the
following fundamental theorem about good random matrices.

\begin{theorem}\label{th:GoodRandomMatrixDuality}
A random matrix is good if and only if its transpose is good.
\end{theorem}

The following lemma will be needed for the proof.

\begin{lemma}\label{le:UniformRandomVector}
  Let $\rseq{v}$ be a random $m$-dimensional vector over
  $\field{q}$. If for all $\seq{u} \in \field{q}^m \setminus
  \{\seq{0}\}$, $\seq{u} \rseq{v}^T$ is uniformly distributed over
  $\field{q}$, then $\rseq{v}$ is uniformly distributed over
  $\field{q}^m$.
\end{lemma}

\begin{proof}
Since $\seq{u} \rseq{v}^T$ is uniformly distributed over $\field{q}$ for all $\seq{u} \in \field{q}^m \setminus \{\seq{0}\}$, we have
\[
\sum_{\seq{v}: \seq{u} \seq{v}^T = 0} P\{\rseq{v} = \seq{v}\} = P\{\seq{u} \rseq{v}^T = 0\} = \frac{1}{q},
\]
so that
\[
\sum_{\seq{u}: \seq{u} \ne \seq{0}} \sum_{\seq{v}: \seq{u} \seq{v}^T = 0} P\{\rseq{v} = \seq{v}\} = \frac{q^m-1}{q}.
\]
On the other hand, we note that
\begin{align*}
\sum_{\seq{u}: \seq{u} \ne \seq{0}} \sum_{\seq{v}: \seq{u} \seq{v}^T = 0} P\{\rseq{v} = \seq{v}\}
&= \sum_{\seq{v} \in \field{q}^m} \sum_{\seq{u}: \seq{u} \ne \seq{0}, \seq{u} \seq{v}^T = 0} P\{\rseq{v} = \seq{v}\} \\
&= (q^m-1) P\{\rseq{v} = \seq{0}\} + (q^{m-1}-1) P\{\rseq{v} \ne \seq{0}\} \\
&= q^{m-1}(q-1) P\{\rseq{v} = \seq{0}\} + q^{m-1}-1.
\end{align*}
Combining these two identities gives $P\{\rseq{v} = \seq{0}\} =
q^{-m}$. Replacing $\rseq{v}$ with $\rseq{v} - \seq{v}$, where $\seq{v}
\in \field{q}^m$(which preserves the uniform distribution of
$\seq{u} \rseq{v}^T$), then gives that $P\{\rseq{v} = \seq{v}\} =
q^{-m}$ for all $\seq{v} \in \field{q}^m$.
\end{proof}


\begin{proof}[Proof of Th.~\ref{th:GoodRandomMatrixDuality}]
  Let $\rmat{A}$ be a random $m\times n$ matrix over $\field{q}$. If
  it is good, then for any $\seq{u} \in \field{q}^m \setminus
  \{\seq{0}\}$ and $\seq{v} \in \field{q}^n \setminus \{\seq{0}\}$,
  the product $\seq{u} \rmat{A} \seq{v}^T = \seq{u} (\seq{v}
  \rmat{A}^T)^T$ is uniformly distributed over
  $\field{q}$. Lemma~\ref{le:UniformRandomVector} shows that $\seq{v}
  \rmat{A}^T$ is uniformly distributed over $\field{q}^m$. In other
  words, $\rmat{A}^T$ is good. Conversely, if $\rmat{A}^T$ is good,
  then $\rmat{A} = (\rmat{A}^T)^T$ is good.
\end{proof}

From Th.~\ref{th:GoodRandomMatrixDuality} it is clear that the
minimum support size of good random $m\times n$ matrices is
$q^{\max\{m,n\}}$. Furthermore, we find a close relation between good
random matrices with minimum support size and MRD codes.

Let $\matset{A}$ be a subset of $\field{q}^{m\times n}$ of size
$|\matset{A}|\geq 2$. The
\emph{rank distance} $\rankd(\matset{A})$ of $\matset{A}$ is the
minimum rank of $\mat{X}-\mat{Y}$ over all distinct $\mat{X}, \mat{Y}
\in \matset{A}$. The \emph{Singleton bound} tells us that the size of
$\matset{A}$ satisfies the following inequality
\cite{Delsarte197800, Gabidulin198501, Roth199103}:
\begin{equation}
|\matset{A}| \le q^{\max\{m,n\}(\min\{m,n\} - \rankd(\matset{A}) + 1)}.
\end{equation}
A set of matrices achieving this upper bound is called a
\emph{maximum-rank-distance (MRD) code}.

\begin{definition}
  Suppose $m,n,k$ are positive integers with $k\leq\min\{m,n\}$.
  An $(m,n,k)$ MRD code over $\field{q}$ is a set $\matset{A}$ of
  $q^{k\max\{m,n\}}$ matrices in $\field{q}^{m\times n}$ such that
  $\rankd(\matset{A}) = \min\{m,n\}-k+1$.
\end{definition}

The next lemma gives an equivalent condition for MRD codes. It is
a consequence of \cite[Th.~5.4]{Delsarte197800}. For
the convenience of our readers we give a direct proof below.

\begin{lemma}\label{le:MRDCode}
  Suppose $m,n,k$ are positive integers with $k \le m \le n$ (resp. $m
  \ge n \ge k$). A set $\matset{A} \subseteq \field{q}^{m\times n}$ is
  an $(m,n,k)$ MRD code if and only if $|\matset{A}| =
  q^{k\max\{m,n\}}$ and $\mat{B}\matset{A} = \field{q}^{k\times n}$
  (resp. $\matset{A}\mat{B} = \field{q}^{m\times k}$) for every
  full-rank matrix $\mat{B} \in \field{q}^{k\times m}$ (resp. $\mat{B}
  \in \field{q}^{n\times k}$).
\end{lemma}

\begin{proof}
(Necessity) It suffices to show that $\mat{B}(\mat{X} - \mat{Y}) =
  \mat{0}$ implies $\mat{X} = \mat{Y}$ for any $\mat{X}, \mat{Y} \in
  \matset{A}$. Note that $\rank(\mat{B}) = k$, so that
\[
\rank(\mat{X}-\mat{Y}) \le \rank(\{\seq{v}\in\field{q}^n:
\mat{B}\seq{v}^T = \seq{0}\}) = m-k,
\]
which implies that $\mat{X} = \mat{Y}$ because $\rankd(\matset{A}) = m-k+1$.

(Sufficiency) It suffices to show that $\rank(\mat{X}-\mat{Y}) \le
m-k$ implies $\mat{X} = \mat{Y}$ for any $\mat{X}, \mat{Y} \in
\matset{A}$. Since $\rank(\mat{X}-\mat{Y}) \le m-k$ implies that there
exists $\mat{B} \in \field{q}^{k\times m}$ such that $\rank(\mat{B}) =
k$ and $\mat{B}(\mat{X}-\mat{Y}) = \mat{0}$, we immediately conclude
$\mat{X} = \mat{Y}$ from the condition $|\matset{A}| =
|\mat{B}\matset{A}| = q^{kn}$.

The case $m \ge n \ge k$ is done similarly.
\end{proof}

The next two theorems establish the relation between good random
matrices and MRD codes.

\begin{theorem}\label{th:GoodRandomMatrixVsMRDCode}
A random matrix uniformly distributed over an $(m,n,k)$ MRD code is good.
\end{theorem}

\begin{proof}
  By Th.~\ref{th:GoodRandomMatrixDuality} we may assume
    $k\leq m\leq n$.  Let $\matset{A}$ be an $(m,n,k)$ MRD code,
    $\seq{u}\in\field{q}^m\setminus\{\seq{0}\}$,
    $\seq{v}\in\field{q}^n$. We must show that
    $|\{\mat{A}\in\matset{A}:\seq{u}\mat{A}=\seq{v}\}|=q^{(k-1)n}$. We
  choose a full-rank matrix $\mat{B}\in\field{q}^{k\times m}$ with
  first row equal to $\seq{u}$. By Lemma~\ref{le:MRDCode} the number
  of $\mat{A}\in\matset{A}$ satisfying $\seq{u}\mat{A}=\seq{v}$ equals
  the number of matrices in $\field{q}^{k\times n}$ with first row
  equal to $\seq{v}$, i.e.\ $q^{(k-1)n}$ as asserted.
\end{proof}

\begin{theorem}\label{th:GoodRandomMatrixWithMinimumSupportSize}
  The minimum support size of a good random $m\times n$ matrix is
  $q^{\max\{m,n\}}$. A random $m\times n$ matrix with support size
  $q^{\max\{m,n\}}$ is good if and only if it is uniformly distributed
  over an $(m,n,1)$ MRD code.
\end{theorem}

\begin{proof}
  The first statement is an easy consequence of
  Theorems~\ref{th:GoodRandomMatrixDuality} and
  \ref{th:GoodRandomMatrixVsMRDCode}. Let us prove the second
  statement.

At first, it follows from Th.~\ref{th:GoodRandomMatrixVsMRDCode}
that a random matrix uniformly distributed over an $(m,n,1)$ MRD code
is a good random matrix with support size
$q^{\max\{m,n\}}$. Conversely, for a good random $m\times n$ matrix
$\rmat{A}$ with $|\supp(\rmat{A})| = q^{\max\{m,n\}}$ and $m \le n$,
we have $|\supp(\rmat{A})| = q^{n}$ and $\seq{u}(\supp(\rmat{A})) =
\field{q}^n$ for every $\seq{u} \in \field{q}^m\setminus\{\seq{0}\}$,
which implies that $\supp(\rmat{A})$ is an $(m,n,1)$ MRD code
(Lemma~\ref{le:MRDCode}) and the probability distribution is
uniform. As for the case of $m>n$, transpose the random matrix and
apply Th.~\ref{th:GoodRandomMatrixDuality}.
\end{proof}

\begin{remark}
  For a good random $m\times n$ matrix $\rmat{A}$ with minimum support
  size, if $\supp(\rmat{A})$ contains the zero matrix, then it follows
  from Th.~\ref{th:GoodRandomMatrixWithMinimumSupportSize} that
\begin{equation}
P\{\rank(\rmat{A}) = \min\{m,n\}\} = 1-q^{-\max\{m,n\}}.
\end{equation}
This is an important property. For comparison, recall that for the
random matrix $\rmat{G}$ defined at the beginning of this paper, we
only have
\[
P\{\rank(\rmat{G}) = \min\{m,n\}\} = \prod_{i=0}^{\min\{m,n\}-1} (1-q^{i-\max\{m,n\}}).
\]
\end{remark}

It is known that $(m,n,k)$ MRD codes exist for all positive
  integers $m,n,k$ with $k\leq\min\{m,n\}$. The standard construction
  is based on the representation of $\field{q}$-linear endomorphisms
  of $\field{q^m}$ by $q$-polynomials of $q$-degree less than $m$
  (polynomials of the form $\sum_{i=0}^{m-1}a_iX^{q^i}$ with
  $a_i\in\field{q^m}$) and is analogous to that of the classical
  Reed-Solomon codes: Take the set $\matset{L}_k$ of all
  $\field{q}$-linear transformations
  $L\colon\field{q^m}\to\field{q^m}$,
  $x\mapsto\sum_{i=0}^{k-1}a_ix^{q^i}$ with $a_i\in\field{q^m}$ (i.e.\
  those represented by $q$-polynomials of $q$-degree less than $k$)
  and form the corresponding set $\matset{A}_k$ of $q^{km}$ matrices
  over $\field{q}$ representing the linear transformations in
  $\matset{L}_k$ with respect to some fixed basis of $\field{q^m}$
  over $\field{q}$. The set $\matset{A}_k$ forms an $(m,m,k)$ MRD
  code. "Rectangular" $(m,n,k)$ MRD codes with $m\geq n$ can then be
  obtained by deleting the last $m-n$ columns, say, of each matrix in
  $\matset{A}_k$ (This follows from Lemma~\ref{le:MRDCode}.)

The construction of MRD codes just described has been found independently in
\cite{Delsarte197800, Gabidulin198501, Roth199103}. In
\cite{Delsarte197800} they were called \emph{Singleton
  systems}. Following recent practice
we refer to them as \emph{Gabidulin codes}.

For the full-rank case
$\rankd(\matset{A})=\min\{m,n\}$ we now restate the construction
in the language of matrix fields, so as to
provide some new insights. Without loss of generality, we suppose
again $m \ge n$.

At first, choose a subring $\matfield{m}$ (containing $1$) of the ring
of all invertible $m\times m$ matrices that is a finite field of order
$q^m$. For example, suppose that $\alpha$ is a primitive element of
$\field{q^m}$ and all elements in $\field{q^m}$ are identified with a
row vector with respect to the basis $\{1, \alpha, \alpha^2, \ldots,
\alpha^{m-1}\}$ of $\field{q^m}$ over $\field{q}$. Writing
$\mat{K}=[\alpha^T\;(\alpha^2)^T\;\hdots\;(\alpha^{m})^T]$ (``%
  companion matrix'' of the minimum polynomial of $\alpha$ over
  $\field{q}$), the set $\matfield{m} = \{\mat{0}, \mat{I}, \mat{K},
\mat{K}^2, \ldots, \mat{K}^{q^m-2}\}$ forms such a
field. The subfield $\matfield{m}$ is clearly an $(m,m,1)$
  MRD code. It is unique up to similarity transformations, i.e\ every
  other subfield of order $q^m$ has the form $\mat{P}^{-1}
  \matfield{m} \mat{P}$ for an appropriate invertible $m\times m$
  matrix $\mat{P}$.

Next, choose an arbitrary $m\times n$ full-rank matrix $\mat{A}$. Then
the linear subspace $\matfield{m} \mat{A}$ of $\field{q}^{m\times
  n}$ is an $(m,n,1)$ MRD code.  We note that to each $\matfield{m}
\mat{A}$ there corresponds an orbit of $\nzmatfield{m}$ on the set
$\Omega(m,n,\field{q})$ of all $m\times n$ full-rank matrices, where
$\nzmatfield{m}$ denotes the multiplicative subgroup of nonzero
elements of $\matfield{m}$ and the action is given by
$(\mat{F},\mat{X}) \mapsto \mat{F}\mat{X}$ for $\mat{F} \in
\nzmatfield{m}$ and $\mat{X} \in \Omega(m,n,\field{q})$. Therefore, for given
$\matfield{m}$, the number of such generated MRD codes is equal to the
number of orbits, i.e.,
\begin{equation}
  \frac{|\Omega(m,n,\field{q})|}{|\nzmatfield{m}|} = \frac{\prod_{i=0}^{n-1}(q^m-q^i)}{q^m-1} = \prod_{i=1}^{n-1}(q^m-q^i).
\end{equation}

\begin{remark}
  \label{re:orthomorphism}
  MRD codes over $\field{q}$ with parameters $(m,2,1)$ (or, mutatis
  mutandis, $(2,n,1)$) are closely related to the concept of an
  orthomorphism and that of a complete mapping of the abelian group
  $(\field{q}^m,+)$. A map $f\colon\field{q}^m\to\field{q}^m$ is said
  to be an \emph{orthomorphism} of $(\field{q}^m,+)$ if both $f$ and $x\mapsto
  f(x)-x$ are bijections of $\field{q}^m$, and a \emph{complete mapping} if $f$
  and $x\mapsto f(x)+x$ are bijections; see
  \cite{evans92,niederreiter-robinson82} for further information.

  An $(m,2,1)$ MRD code $\matset{A}$ over $\field{q}$ with $m\geq 2$
  can be represented in the form
  $\matset{A}=\bigr\{(\seq{x}|f(\seq{x}));\seq{x}\in\field{q}^m\bigr\}$
  with a unique function $f\colon\field{q}^m\to\field{q}^m$. The MRD
  code property requires further that $f$ is a bijection and
  $f(\seq{x})-f(\seq{y})\neq\alpha(\seq{x}-\seq{y})$ for all
  $\seq{x},\seq{y}\in\field{q}^m$ with $\seq{x}\neq\seq{y}$ and all
  $\alpha\in\field{q}$; in other words, the maps $x\mapsto
  f(x)+\alpha x$ must be bijections of $\field{q}^m$ for all
  $\alpha\in\field{q}$. In the special case $q=2$ MRD codes with
  parameters $(m,2,1)$, orthomorphisms of $\field{2}^m$ and complete
  mappings of $\field{2}^m$ are thus equivalent.  This observation
  implies in particular that there exist binary $(m,2,1)$ MRD codes
  which are neither linear nor equal to a coset of a linear code (on
  account of the known existence of nonlinear complete mappings of
  $\field{2}^m$ for $m=4$; see \cite{yuan-tong-zhang07}).
\end{remark}

Th.~\ref{th:GoodRandomMatrixWithMinimumSupportSize} together with
the above comments gives the construction of good random matrices with
minimum support size. For the construction of a general good random
matrix, we need the following simple property.

\begin{theorem}\label{th:GeneralConstruction}
  Let $\rmat{A}$ be a good random $m\times n$ matrix and $\rmat{B}$ an
  arbitrary random $m\times n$ matrix. Let $\rmat{P}$ be a random
  invertible $m\times m$ matrix and $\rmat{Q}$ a random invertible
  $n\times n$ matrix. If $\rmat{A}$ is independent of $(\rmat{P},
  \rmat{Q}, \rmat{B})$, then $\rmat{P} \rmat{A} \rmat{Q} + \rmat{B}$
  is good.
\end{theorem}

The proof is left to the reader.

In Th.~\ref{th:GeneralConstruction}, if $\supp(\rmat{A})$ is a
subspace of $\field{q}^{m\times n}$ and $\rmat{P}$ and $\rmat{Q}$ are
simply identity matrices, then
\[
|\supp(\rmat{A} + \rmat{B})| = |\supp(\rmat{A})|
|\pi_{\supp(\rmat{A})}(\supp(\rmat{B}))|,
\]
where $\pi_{\supp(\rmat{A})}$ denotes the canonical projection from
$\field{q}^{m\times n}$ onto $\field{q}^{m\times
  n}/\supp(\rmat{A})$. Therefore, based on good random matrices with
minimum support size, it is possible to construct a good random matrix
whose support size is a multiple of $q^{\max\{m,n\}}$.

\section{$k$-Good Random Matrices}
\label{sec:KGoodRandomMatrix}

In Section~\ref{sec:GoodRandomMatrix} we showed a one-to-one
correspondence between good random $m\times n$ matrices with minimum
support size and $(m,n,1)$ MRD codes. From the viewpoint of
aesthetics, the current ``picture'' of good random matrices, MRD codes
and their relations is not satisfactory, since there is yet no
appropriate place for a general $(m,n,k)$ MRD code with $k > 1$. In
this section we study properties of random matrices uniformly
distributed over general $(m,n,k)$ MRD codes and round off the picture. 

\begin{definition}\label{df:KGoodRandomMatrix}
  Let $k,m,n$ be positive integers with $k \le \min\{m,n\}$. A random
  $m\times n$ matrix $\rmat{A}$ is said to be \emph{$k$-good} if for every
  full-rank matrix $\mat{M} \in \field{q}^{k\times m}$, the product
  $\mat{M} \rmat{A}$ is uniformly distributed over $\field{q}^{k\times
    n}$. The collection of all $k$-good random $m\times n$ matrices is
  denoted by $\grmat_k(m,n,\field{q})$.
\end{definition}
Thus a random $m\times n$ matrix $\rmat{A}$ is
  $k$-good if for any choice of $k$ linearly independent vectors $\seq{u}_1,
  \ldots, \seq{u}_k \in \field{q}^m$ and any choice of $k$ arbitrary
  vectors $\seq{v}_1, \ldots, \seq{v}_k \in \field{q}^n$ we have
\[
P\{\seq{u}_1 \rmat{A} = \seq{v}_1, \seq{u}_2 \rmat{A} = \seq{v}_2, \ldots, \seq{u}_k \rmat{A} =
\seq{v}_k\} = q^{-kn};
\]
equivalently, the random vectors
$\seq{u}_1\rmat{A},\seq{u}_2\rmat{A},\dots,\seq{u}_k\rmat{A}\in\field{q}^n$ are
uniformly distributed and independent.

In the language of linear transformations,
Def.~\ref{df:KGoodRandomMatrix} requires that for every
$k$-dimensional subspace $U\subseteq\field{q}^m$ the restriction map
$\Hom(\field{q}^m,\field{q}^n)\to\Hom(U,\field{q}^n)$, $f\mapsto f|_U$
induces the uniform distribution on $\Hom(U,\field{q}^n)$. Here
$\Hom(V,W)$ denotes the space of all linear transformations from $V$
to $W$.

In the sequel, when speaking of a $k$-good random $m\times n$ matrix, we shall tacitly assume that $k \le \min\{m,n\}$. By definition, it is clear that
\[
\grmat_{\min\{m,n\}}(m,n,\field{q}) \subseteq \cdots \subseteq \grmat_{2}(m,n,\field{q}) \subseteq \grmat_{1}(m,n,\field{q}) = \grmat(m,n,\field{q}).
\]

Our goal is now to generalize the one-to-one correspondence of
Th.~\ref{th:GoodRandomMatrixWithMinimumSupportSize} to $k$-good
random $m\times n$ matrices and $(m,n,k)$ MRD codes. To this end, we
shall first establish a generalization of
Th.~\ref{th:GoodRandomMatrixDuality}.

\begin{theorem}\label{th:KGoodRandomMatrixDuality}
A random $m\times n$ matrix is $k$-good if and only if its transpose
is $k$-good. 
\end{theorem}

Analogous to the proof of Th.~\ref{th:GoodRandomMatrixDuality}, we
need a generalization of Lemma~\ref{le:UniformRandomVector}. But
before that, we need some preparatory counting lemmas. 

\begin{lemma}[{\cite[Th.~3.3]{Hirschfeld199800}}]
\label{le:GeneralizedAnzahlTheorem}
Let $k,l,m,n$ be positive integers with $k \le \min\{l,m\}$ and $l+m-k
\le n$. Let $M$ be an $m$-dimensional subspace of
$\field{q}^n$. Define the set 
\[
\mathcal{L} \eqdef \{L: \mbox{$L$ is an $l$-dimensional subspace of
  $\field{q}^n$ such that $\dim(L \cap M) = k$}\}.
\]
Then
\[
|\mathcal{L}| = q^{(l-k)(m-k)} \gbinom{m}{k} \gbinom{n-m}{l-k}.
\]
\end{lemma}

\begin{lemma}\label{le:SpecialAnzahlLemma}
  Let $k, l, m, n$ be positive integers with $k \le l \le m \le n$ and
  $k+n \ge l+m$. Let $\mat{N}$ be an $m\times n$ matrix of rank
  $l$. Define the set
\[
\matset{M} \eqdef \{\mat{M} \in \field{q}^{m\times n}: \rank(\mat{M})
= m, \rank(\mat{M} \mat{N}^T) = k\}.
\]
Then
\[
|\mathcal{M}| = q^{k(n-l-m+k)} \gbinom{l}{k} \gbinom{n-l}{m-k}
\left[\prod_{i=0}^{m-1} (q^m-q^i)\right].
\]
\end{lemma}

\begin{proof}
  Define mappings $f: \field{q}^m \to \field{q}^n$,
  $\seq{u} \mapsto \seq{u}\mat{M}$ and $g: \field{q}^n \to
  \field{q}^m$, $\seq{v} \mapsto \seq{v}\mat{N}^T$. By
  assumption we have $\dim(f(\field{q}^m)) = m$ and
  $\dim(g(\field{q}^n)) = l$. Our goal is to count the number of
  mappings $f$ such that $\dim(g(f(\field{q}^m))) = k$. Since
  $g(f(\field{q}^m)) \cong f(\field{q}^m) / (f(\field{q}^m) \cap \ker
  g)$ and $\dim(\ker g) = n-\dim(g(\field{q}^n)) = n-l$,
  Lemma~\ref{le:GeneralizedAnzahlTheorem} shows that the number of
  possible image spaces $f(\field{q}^m)$ is
\[
q^{k(n-l-m+k)} \gbinom{l}{k} \gbinom{n-l}{m-k},
\]
so that the number of all possible mappings $f$ is
\[
q^{k(n-l-m+k)} \gbinom{l}{k} \gbinom{n-l}{m-k}
\left[\prod_{i=0}^{m-1} (q^m-q^i)\right].
\]
\end{proof}

We are now ready to state and prove the generalization of
Lemma~\ref{le:UniformRandomVector}, which will be used in the proof of
Th.~\ref{th:KGoodRandomMatrixDuality}.

\begin{lemma}\label{le:UniformRandomMatrix}
A random $m\times n$ matrix $\rmat{N}$ with $m \le n$ is uniformly distributed over $\field{q}^{m\times n}$ if and only if for every full-rank matrix $\mat{M} \in \field{q}^{m\times n}$, the product $\mat{M}\rmat{N}^T$ is uniformly distributed over $\field{q}^{m\times m}$.
\end{lemma}

\begin{proof}
(Necessity) Given any full-rank matrix $\mat{M} \in \field{q}^{m\times n}$, it is clear that the set $\{\mat{N} \in \field{q}^{m\times n}: \mat{M}\mat{N}^T = \mat{0}\}$ is an $m(n-m)$-dimensional subspace of $\field{q}^{m\times n}$. Consequently, $P\{\mat{M}\rmat{N}^T = \mat{0}\} = q^{m(n-m)} q^{-mn} = q^{-m^2}$. Finally, for any $\mat{K} \in \field{q}^{m\times m}$, there exists a matrix $\mat{N} \in \field{q}^{m\times n}$ such that $\mat{M}\mat{N}^T = \mat{K}$, so we have $P\{\mat{M}\rmat{N}^T = \mat{K}\} = P\{\mat{M}(\rmat{N}-\mat{N})^T = \mat{0}\} = P\{\mat{M}\rmat{N}^T = \mat{0}\} = q^{-m^2}$.

(Sufficiency) Since $\mat{M}\rmat{N}^T$ is uniformly distributed over $\field{q}^{m\times m}$ for any full-rank matrix $\mat{M} \in \field{q}^{m\times n}$, we have
\[
\sum_{\mat{M}: \rank(\mat{M})=m} P\{\rank(\mat{M}\rmat{N}^T) = k\}
= q^{-m^2} \gbinom{m}{k} \left[\prod_{i=0}^{k-1} (q^m-q^i)\right] \left[\prod_{i=0}^{m-1} (q^n-q^i)\right].
\]
On the other hand, the sum can be rewritten as
\begin{align*}
\sum_{\mat{M}: \rank(\mat{M})=m} P\{\rank(\mat{M}\rmat{N}^T) = k\}
&= \sum_{\mat{M}: \rank(\mat{M})=m} \sum_{\substack{\mat{N}: \rank(\mat{N}) \ge k,\\ \rank(\mat{M}\mat{N}^T)=k}} P\{\rmat{N} = \mat{N}\} \\
&= \sum_{l=k}^{\min\{m, n-m+k\}}
\sum_{\substack{(\mat{M}, \mat{N}): \rank(\mat{M}) = m,\\
\rank(\mat{N}) = l,\\
\rank(\mat{M}\mat{N}^T)=k}} P\{\rmat{N} = \mat{N}\}.
\end{align*}
It follows from Lemma~\ref{le:SpecialAnzahlLemma} that
\[
\begin{split}
\sum_{\substack{(\mat{M}, \mat{N}): \rank(\mat{M}) = m,\\
\rank(\mat{N}) = l,\\
\rank(\mat{M}\mat{N}^T)=k}} P\{\rmat{N} = \mat{N}\}
&= q^{k(n-l-m+k)} \gbinom{l}{k} \gbinom{n-l}{m-k} \left[\prod_{i=0}^{m-1} (q^m-q^i)\right] \\
&\quad \times P\{\rank(\rmat{N}) = l\}.
\end{split}
\]
Combining the above identities gives
\[
\begin{split}
&\sum_{l=k}^{\min\{m, n-m+k\}}
q^{-k(l-k)} \gbinom{l}{k} \gbinom{n-l}{m-k}
P\{\rank(\rmat{N}) = l\} \\
&\qquad = q^{-mn} q^{(m-k)(n-m)} \gbinom{n}{m} \gbinom{m}{k} \left[\prod_{i=0}^{k-1} (q^m-q^i)\right]
\qquad \text{for $k=0,1,\ldots,m$}.
\end{split}
\]
Solving these equations with the identities
\[
\sum_{l=k}^{\min\{m, n-m+k\}} \gbinom{n-m}{l-k}
\left[\prod_{i=0}^{l-k-1} (q^{m-k}-q^i)\right] 
= q^{(m-k)(n-m)}
\]
and
\[
\gbinom{n}{l} \gbinom{l}{k} \gbinom{n-l}{m-k} = \gbinom{n}{m}
\gbinom{m}{k}\gbinom{n-m}{l-k}, 
\]
we obtain
\[
P\{\rank(\rmat{N}) = l\} = q^{-mn} \gbinom{n}{l}
\left[\prod_{i=0}^{l-1} (q^m-q^i)\right] \qquad \text{for $l = 0, 1,
  \ldots, m$}. 
\]
In particular, $P\{\rmat{N} = \mat{0}\} = q^{-mn}$. Replacing
$\rmat{N}$ with $\rmat{N} - \mat{N}$ 
then shows that $P\{\rmat{N} = \mat{N}\} =
q^{-mn}$ for all $\mat{N} \in \field{q}^{m\times n}$.
\end{proof}

\begin{proof}[Proof of Th.~\ref{th:KGoodRandomMatrixDuality}]
  Let $\rmat{A}$ be a $k$-good random $m\times n$ matrix. Then for any
  full-rank matrix $\mat{M} \in \field{q}^{k \times m}$, the product
  $\mat{M} \rmat{A}$ is uniformly distributed over $\field{q}^{k\times
    n}$. From Lemma~\ref{le:UniformRandomMatrix}, it follows that for
  any full-rank matrix $\mat{N} \in \field{q}^{k\times n}$, the
  product $(\mat{M} \rmat{A}) \mat{N}^T = \mat{M} (\mat{N}
  \rmat{A}^T)^T$ is uniformly distributed over $\field{q}^{k\times
    k}$. Lemma~\ref{le:UniformRandomMatrix} further shows that
  $\mat{N} \rmat{A}^T$ is uniformly distributed over
  $\field{q}^{k\times m}$. In other words, $\rmat{A}^T$ is
  $k$-good. Conversely, if $\rmat{A}^T$ is $k$-good then $\rmat{A} =
  (\rmat{A}^T)^T$ is $k$-good.
\end{proof}

The next theorem, a generalization of
Th.~\ref{th:GoodRandomMatrixWithMinimumSupportSize}, gives the
one-to-one correspondence between $k$-good random matrices of
minimum support size and $(m,n,k)$ MRD codes.

\begin{theorem}\label{th:KGoodRandomMatrixWithMinimumSupportSize}
  The minimum support size of a $k$-good random $m\times n$ matrix is
  $q^{k\max\{m,n\}}$. A random $m\times n$ matrix with support size
  $q^{k\max\{m,n\}}$ is $k$-good if and only if it is uniformly
  distributed over an $(m,n,k)$ MRD code.
\end{theorem}

\begin{proof}
  Owing to Th.~\ref{th:KGoodRandomMatrixDuality}, it suffices to
  establish the theorem for $m \le n$. It is clear that the support
  size of a $k$-good random $m\times n$ matrix must not be smaller
  than $q^{kn} = q^{k\max\{m,n\}}$. Lemma~\ref{le:MRDCode} shows that
  a random $m\times n$ matrix uniformly distributed over an $(m,n,k)$
  MRD code achieves this lower bound. On the other hand, for a
  $k$-good random $m\times n$ matrix $\rmat{A}$ with
  $|\supp(\rmat{A})| = q^{k\max\{m,n\}}$, we have $|\supp(\rmat{A})| =
  q^{kn}$ and $\mat{M}(\supp(\rmat{A})) = \field{q}^{k\times n}$ for
  every full-rank matrix $\mat{M} \in \field{q}^{k\times m}$, which
  implies that $\supp(\rmat{A})$ is an $(m,n,k)$ MRD code
  (Lemma~\ref{le:MRDCode}) and the probability distribution is
  uniform.
\end{proof}

A corollary follows immediately.

\begin{corollary}
  The random matrix uniformly distributed over $\field{q}^{m\times n}$
  is the unique $\min\{m,n\}$-good random $m\times n$ matrix (up to
  probability distribution).
\end{corollary}

There is also a generalization of Th.~\ref{th:GeneralConstruction}
for constructing a general $k$-good random matrix.

\begin{theorem}\label{th:KGeneralConstruction}
  Let $k,m,n,s,t$ be positive integers such that $k \le s \le m$ and
  $k \le t \le n$. Let $\rmat{A}$ be a $k$-good random $m\times n$
  matrix and $\rmat{B}$ an arbitrary random $s\times t$ matrix. Let
  $\rmat{P}$ be a random full-rank $s\times m$ matrix and $\rmat{Q}$ a
  random full-rank $n\times t$ matrix. If $\rmat{A}$ is independent of
  $(\rmat{P}, \rmat{Q}, \rmat{B})$, then $\rmat{P} \rmat{A} \rmat{Q} +
  \rmat{B}$ is a $k$-good random $s\times t$ matrix.
\end{theorem}

The theorem readily follows from the observation that, under the
  given conditions, for any full-rank matrices $\mat{M}\in
  \field{q}^{k\times s}$, $\mat{P}\in
  \field{q}^{s\times m}$, $\mat{Q}\in
  \field{q}^{n\times t}$ the product $\mat{MP}\rmat{A}\mat{Q}$ is
  uniformly distributed over $\field{q}^{k\times t}$. The details are
  left to the reader.

\section{Homogeneous Weights on Matrix Spaces}\label{sec:whom}

Further examples of $k$-good random matrices are provided by the
so-called left and right homogeneous weights on $\field{q}^{m\times n}$, suitable
scaled to turn it into a probability distribution. Denoting by $R_t$
the ring of $t\times t$ matrices over $\field{q}$, the set 
$\field{q}^{m\times n}$ can be regarded as an $R_m$-$R_n$ bimodule,
i.e.\ it is both a left $R_m$-module and a right $R_n$-module
(relative to the natural actions) and has the property
that $\mat{AXB}=(\mat{AX})\mat{B}=\mat{A}(\mat{XB})$ is independent of
the choice of parenthesis. The left homogeneous weight
$\lhom\colon\field{q}^{m\times n}\to\mathbb{R}$ is uniquely defined by the
following axioms:
\begin{enumerate}[(H1)]
\item $\lhom(\mat{0})=0$;
\item $\lhom(\mat{UX})=\lhom(\mat{X})$ for all
  $\mat{X}\in\field{q}^{m\times n}$, $\mat{U}\in
  R_m^\times$;\footnotemark
\item $\sum_{\mat{X}\in\matset{U}}\lhom(\mat{X})=|\matset{U}|$
  for all cyclic left submodules $\matset{U}\neq\{\mat{0}\}$ of
  $\field{q}^{m\times n}$.
\end{enumerate}
\footnotetext{By $R_m^\times$ we denote the group of units of $R_m$,
  i.e.\ the general linear group of degree $m$ over $\field{q}$.}
This represents the more general definition of a homogeneous
weight in \cite{greferath-schmidt99}, adapted to the case of modules
over finite rings considered in \cite{st:homippi}. 
According to \cite[Prop.~7]{st:homippi} the function $\lhom$ is
explicitly given by
\begin{equation}
  \label{eq:lhom}
  \lhom(\mat{X})=1-\frac{(-1)^{\rank\mat{X}}}{(q^m-1)(q^{m-1}-1)
    \dotsm(q^{m-\rank\mat{X}+1}-1)}.
\end{equation}
The left submodules of $\field{q}^{m\times n}$ are in one-to-one
correspondence with the subspaces of $\field{q}^n$ by the map 
sending a left submodule $\matset{U}$ to the subspace of $\field{q}^n$
generated by the row spaces of all matrices in $\matset{U}$; see
\cite[Lemma~1]{st:homippi} for example.
In the case $m\geq n$ all left submodules of $\field{q}^{m\times n}$
are cyclic, so that the equation in (H3) holds for all left submodules
$\matset{U}$ of $\field{q}^{m\times n}$ and $\lhom$ is a left
homogeneous weight on $\field{q}^{m\times n}$ in the stronger sense
defined in \cite{st:homippi}. If $m<n$ then $\field{q}^{m\times n}$
contains non-cyclic left submodules (those which correspond to
subspaces of $\field{q}^n$ of dimension $>m$) and the equation
in (H3) does not remain true for these.

The right homogeneous weight $\rhom\colon\field{q}^{m\times
  n}\to\mathbb{R}$ is defined in the analogous way using the right
$R_n$-module structure of $\field{q}^{m\times n}$. The preceding
remarks hold mutatis mutandis for $\rhom$. From \eqref{eq:lhom} it is
clear that $\lhom\neq\rhom$ in the ``rectangular'' case $m\neq n$
(while of course $\lhom=\rhom$ for $m=n$).

Obviously $\lhom$ (and similarly $\rhom$) can be scaled by a constant
$\gamma>0$ to turn it into a probability distribution on
$\field{q}^{m\times n}$. The normalized version $\nlhom=\gamma\lhom$ satisfies
(H1), (H2), and $\sum_{\mat{X}\in\field{q}^{m\times
    n}}\lhom(\mat{X})=\gamma|\matset{U}|$ for all cyclic left
submodules $\matset{U}\neq\{\mat{0}\}$ of $\field{q}^{m\times n}$ in
place of (H3). The constant is $\gamma=c_{mn}^{-1}$, where
$c_{mn}=\sum_{\mat{X}\in\field{q}^{m\times n}}\lhom(\mat{X})$ is the
total left homogeneous weight of $\field{q}^{m\times n}$.

\begin{lemma}
  \label{le:totalweight}
  For positive integers $m,n$ the total left homogeneous weight of
  $\field{q}^{m\times n}$ is $c_{mn}=q^{mn}-(-1)^mq^{m(m+1)/2}\gbinom{n-1}{m}$.
\end{lemma}
\begin{proof}
  For $0\leq r\leq\min\{m,n\}$ the number of rank $r$ matrices in
  $\field{q}^{m\times n}$ equals
  $\gbinom{m}{r}\gbinom{n}{r}\prod_{i=0}^{r-1}(q^r-q^i)
  =q^{r(r-1)/2}\gbinom{n}{r}\prod_{i=0}^{r-1}(q^{m-i}-1)$.
Together with
  \eqref{eq:lhom} this gives
  \begin{align*}
    c_{mn}&=q^{mn}-\sum_{r=0}^m(-1)^rq^{r(r-1)/2}\gbinom{n}{r}\\
    &=q^{mn}-(-1)^mq^{m(m+1)/2}\gbinom{n-1}{m},
  \end{align*}
  where the last step follows from the $q$-analogue of a well-known
  identity for binomial coefficients (the case $q=1$).\footnote{The
    identity is easily proved by expanding
    $(1+X)^{-1}\prod_{i=0}^{n-1}(1+q^iX)$ in two different ways using
    the $q$-binomial theorem \cite{andrews98}.}
\end{proof}
Our particular interest in homogeneous weights in this
paper is due to the following result:
\begin{theorem}
  \label{thm:hom}
  If $m\geq n$ then the normalized left homogeneous weight $\nlhom$
  defines a $k$-good random matrix on $\field{q}^{m\times n}$ for
  $1\leq k\leq n-1$. Similarly, if $m\leq n$ then $\nrhom$
  defines a $k$-good random matrix on $\field{q}^{m\times n}$ for
  $1\leq k\leq m-1$.
\end{theorem}
It follows from \cite[Th.~2]{wt:egal} that the normalized homogeneous weight on
a finite Frobenius ring $R$ (normalized in such a way that it forms a
probability distribution) induces the uniform distribution on the
coset space $R/I$ for every left or right ideal $I\neq\{0\}$ of
$R$. Applying this to the matrix rings $R=R_m$ provides the key step
in the proof of the theorem for $m=n$. Here we use similar ideas to prove 
the theorem in the general case. 
\begin{proof}[Proof of Th.~\ref{thm:hom}]
  Clearly $\nlhom\colon\field{q}^{m\times n}\to\mathbb{R}$ and
  $\nrhom\colon\field{q}^{n\times m}\to\mathbb{R}$ are related by
  $\nlhom(\mat{X})=\nrhom(\mat{X}^T)$ for all $\mat{X}\in\field{q}^{m\times
    n}$. Hence by Th.~\ref{th:KGoodRandomMatrixDuality} it
  suffices to prove the second assertion. So we assume from now on
  that $m\leq n$.
  
  The weight $\nrhom\colon\field{q}^{m\times n}\to\mathbb{R}$ gives rise to a
  $k$-good random matrix if and only if for every
  $\mat{B}\in\field{q}^{k\times m}$ with $\rank(\mat{B})=k$ and every
  $\mat{Y}\in\field{q}^{k\times n}$ the equation
  \begin{equation*}
    \sum_{\substack{\mat{X}\in\field{q}^{m\times n}\\\mat{BX}=\mat{Y}}}\nrhom(\mat{X})=q^{-kn}
  \end{equation*}
holds.
The equation $\mat{BX}=\mat{0}$ is equivalent to the statement that
the column space of $\mat{X}$ is contained in the orthogonal space
$U=V^\perp$ of the row space $V$ of $\mat{B}$. Since
$\rank(\mat{B})=k$, we have $\dim(U)=m-k$. Hence
$\matset{U}=\{\mat{X}\in\field{q}^{m\times n};\mat{BX}=\mat{0}\}$ is a
right submodule of $\field{q}^{m\times n}$ of size
$|\matset{U}|=q^{(m-k)n}$, and we have $\matset{U}\neq\{\mat{0}\}$
provided that $1\leq k\leq m-1$. Hence in order to complete the proof,
it suffices to show that $\nrhom$ not only satisfies the analogue of
(H3) but the stronger
\begin{equation}
  \label{thm:hom:proof}
  \sum_{\mat{X}\in\matset{U}+\mat{A}}\nrhom(\mat{X})
  =q^{-mn}|\matset{U}+\mat{A}|=q^{-kn}
\end{equation}
for every coset $\matset{U}+\mat{A}$ of every (cyclic) right submodule
$\matset{U}\neq\{\mat{0}\}$ of $\field{q}^{m\times n}$; in other
words, that $\nrhom$ induces the uniform distribution on the
coset space of every nonzero right submodule of $\field{q}^{m\times n}$.

For the proof of this fact we proceed as follows: Given $\matset{U}$,
we define a weight $w$ on the quotient module
$M=\field{q}^{m\times n}/\matset{U}$ by
$w(\matset{U}+\mat{A})=\sum_{\mat{X}\in\matset{U}+\mat{A}}\nrhom(\mat{X})$.
It is a standard result from ring theory that
$M\cong\field{q}^{k\times  n}$ as right $R_n$-modules. Now we show that
$w$ satisfies the analogues of (H2), (H3), i.e.\
  $w(x)=w(x\mat{U})$ for all $x\in M$, $\mat{U}\in R_n^\times$ and
  $\sum_{x\in N}w(x)=q^{-kn}|N|$ for all (cyclic) submodules $N\neq\{0\}$
  of $M$, and in place of
(H1) the equation $w(0)=q^{-kn}$. Then we appeal to
\cite[Prop.~5]{st:homippi}, which says that the space of all functions
$f\colon M\to\mathbb{R}$ satisfying the
analogues of (H2), (H3) is generated by $\nrhom\colon M\to\mathbb{R}$
and the uniform distribution $M\to\mathbb{R}$,
$\matset{U}+\mat{A}\mapsto q^{-kn}$. So there exist
$\alpha,\beta\in\mathbb{R}$ such that $w(x)=\alpha\nrhom(x)+\beta
q^{-kn}$ for all $x=\matset{U}+\mat{A}\in M$. Setting $x=0$ gives
$\beta=1$ (since $\nrhom(0)=0$). This in turn implies $\alpha=0$, 
since both $w$ and $\nrhom$ are probability distributions. Thus
$w(x)=q^{-kn}$ for all $x\in M$ and we are done.
\end{proof}
\begin{example}
  First we consider the case of binary $2\times 3$ matrices. The space
  $\field{2}^{2\times 3}$ contains $21$ matrices of rank $1$ (parametrized as
  $\seq{u}^T\seq{v}$ with $\seq{u}\in\field{2}^2\setminus\{\seq{0}\}$,
   $\seq{v}\in\field{2}^3\setminus\{\seq{0}\}$) and $42$ matrices of
  rank $2$. The normalized left and
  right homogeneous weights $\nlhom$, $\nrhom$ on $\field{2}^{2\times 3}$ are
  given by the following tables:
  \begin{equation}
    \label{eq:2x3}
    \begin{array}{c|ccc}
      \rank(\mat{X})&0&1&2\\\hline
      \nlhom(\mat{X})&0&\frac{1}{42}&\frac{1}{84}
    \end{array}\qquad
    \begin{array}{c|ccc}
      \rank(\mat{X})&0&1&2\\\hline
      \nrhom(\mat{X})&0&\frac{1}{56}&\frac{5}{336}
    \end{array}
  \end{equation}
The weight $\nlhom$ is a probability distribution on
$\field{2}^{2\times 3}$ and satisfies (H1), (H2). However, since
$|\matset{U}|^{-1}\sum_{\mat{X}\in\matset{U}}\nlhom(\mat{X})=\frac{1}{56}$
 for all submodules
$\matset{U}\neq\{\mat{0}\}$ and $\matset{U}\neq\field{2}^{2\times
  3}$, the weight $\nlhom$ cannot yield a $1$-good random
$2\times 3$ matrix over $\field{2}$.

On the other hand, the weight $\nrhom$ defines, 
by Th.~\ref{thm:hom},
a $1$-good random matrix $\rmat{A}\in\field{2}^{2\times
  3}$. This means that every coset of a right submodule $\matset{U}$
of $\field{2}^{2\times 3}$ of size $|\matset{U}|=8$ (which is one of
the modules $\matset{U}_1$, $\matset{U}_2$, $\matset{U}_3$
corresponding to column spaces generated by $\left(
  \begin{smallmatrix}
    1\\0
  \end{smallmatrix}\right)$,
 $\left(
  \begin{smallmatrix}
    0\\1
  \end{smallmatrix}\right)$,
 $\left(
  \begin{smallmatrix}
    1\\1
  \end{smallmatrix}\right)$, respectively)
has total weight $1/8$. For the submodules $\matset{U}_i$ this is
obvious, since they contain the all-zero $2\times 3$ matrix
and $7$ matrices of rank $1$ and weight $1/56$. For the remaining cosets
$\matset{U}_i+\mat{A}$ with $\mat{A}\notin\matset{U}_i$ it implies that
each such coset contains $2$ matrices of rank $1$ and $6$ matrices of
rank $2$.

Then we consider the case of binary $3\times 2$ matrices. By
Th.~\ref{thm:hom} (or direct application of
Th.~\ref{th:KGoodRandomMatrixDuality}) the weight
$\nlhom\colon\field{2}^{3\times 2}\to\mathbb{R}$, which is the
transpose of $\nrhom\colon\field{2}^{2\times 3}\to\mathbb{R}$, defines
a $1$-good random matrix on $\field{2}^{3\times 2}$. Since
$1$-goodness refers to the right module structure of
$\field{2}^{3\times 2}$, this provides additional insight into the
combinatorics of the rank function; namely, one may infer the
following rank distributions of the cosets of right submodules
$\matset{U}$ of $\field{2}^{3\times 2}$ of size
$|\matset{U}|=16$. Each such submodule contains $1$ matrix of rank
$0$, $9$ matrices of rank $1$ and $6$ matrices of rank $2$. Each of
the three cosets $\matset{U}+\mat{A}$ with $\mat{A}\notin\matset{U}$
contains $4$ matrices of rank $1$ and $12$ matrices of rank $2$.
\end{example}

\section{Applications of $k$-Good Random 
Matrices}\label{sec:applic}

As shown at the beginning of this paper, $1$-good random matrices are
of fundamental importance in information theory. Quite naturally one may
wonder whether there are any applications of $k$-good random
matrices with $k>1$. In this section we
present 
such applications 
and their relations to some well-known combinatorial
problems. Some of the proofs are easy and hence left to the reader.

The following result is a direct consequence of
  Def.~\ref{df:KGoodRandomMatrix} and the fact that distinct
  nonzero vectors in a binary vector space are linearly independent.

\begin{proposition}\label{pr:Application1}
  Let $\rmat{A}$ be a $2$-good random matrix over $\field{2}$. The
  random mapping $F: \field{2}^m\setminus\{\seq{0}\} \to \field{2}^n$
  given by $\seq{u} \mapsto \seq{u}\rmat{A}$ satisfies the
  pairwise-independence condition:
\[
P\{F(\seq{u}_1) = \seq{v}_1, F(\seq{u}_2) = \seq{v}_2\} =
P\{F(\seq{u}_1) = \seq{v}_1\} P\{F(\seq{u}_2) = \seq{v}_2\} = 2^{-2n},
\]
where $\seq{u}_1, \seq{u}_2 \in \field{2}^m \setminus\{\seq{0}\}$,
$\seq{v}_1, \seq{v}_2 \in \field{2}^n$, and $\seq{u}_1 \ne \seq{u}_2$.
\end{proposition}

Pairwise-independence of codewords is a prerequisite in the
direct-part proof of the coding theorem for lossless JSCC (also
including the case of channel coding), so
Prop.~\ref{pr:Application1} indeed provides an alternative
random coding scheme for optimal lossless JSCC
(cf. \cite[Sec. III-C]{Yang200904}).

\begin{definition}[see \cite{Cohen199411}]
  A set of vectors in $\field{q}^n$ is said to be \emph{intersecting}
  if there is at least one position $i$ such that all their $i$-th
  components are nonzero. A linear code is said to be \emph{$k$-wise
    intersecting} if any $k$ linearly independent codewords
  intersect.
\end{definition}

\begin{proposition}\label{pr:Application2}
  Let $\rmat{A}$ be a $k$-good random $m\times n$ matrix. Define the
  random set $\rset{C} \eqdef \{\seq{u}\rmat{A}: \seq{u} \in
  \field{q}^m\}$. Then we have
\[
P\{\text{$\rset{C}$ is not $k$-wise intersecting}\} \le
\left(1-(1-q^{-1})^k\right)^n \prod_{i=0}^{k-1} (q^m-q^i).
\]
Asymptotically, as $m,n$ go to infinity with $m/n = r$,
\[
P\{\text{$\rset{C}$ is not $k$-wise intersecting}\} =
\uorder\left(q^{n\left[(r-1)k + \log_q(q^k - (q-1)^k)\right]}\right),
\]
so that whenever
\begin{equation}\label{eq:IntersectingCodeBound}
  r < 1 - \frac{1}{k} \log_q(q^k - (q-1)^k),\qquad
  (\text{cf. \cite[Th.~3.2]{Cohen199411}})
\end{equation}
$\rset{C}$ is $k$-wise intersecting with high probability that
converges to $1$ as $n \to \infty$. 
\end{proposition}

\begin{proof}[Sketch of Proof]
Note that the expectation
\[
\begin{split}
E\Big[\Big|\Big\{(\seq{v}_i)_{i=1}^k \in (\field{q}^n)^k:{}
&\text{$\seq{v}_i \in \rset{C}$ for $i=1,2,\ldots, k$ and
  $(\seq{v}_i)_{i=1}^k$ are}\\
&\text{linearly independent and not intersecting}\Big\}\Big|\Big]
\end{split}
\]
is slightly overestimated by
\[
\begin{split}
E\Big[\Big|\Big\{(\seq{u}_i)_{i=1}^k \in (\field{q}^m)^k:{}
&\text{$(\seq{u}_i)_{i=1}^k$ are linearly independent and}\\
&\text{$(\seq{u}_i\rmat{A})_{i=1}^k$ are not intersecting}\Big\}\Big|\Big].
\end{split}
\]
The latter can be easily computed by the property of a $k$-good random
matrix. The proof is completed by applying Markov's inequality. 
\end{proof}

Proposition~\ref{pr:Application2} shows that a $k$-good random matrix
can achieve the asymptotic (random coding) lower bound \eqref{eq:IntersectingCodeBound}
of maximum rate of linear $k$-wise intersecting codes. Recall that linear $k$-wise
intersecting codes has a close relation to many problems in
combinatorics, such as separating systems \cite{Bose198007,
  Pradhan197609}, qualitative independence \cite{Cohen199411},
frameproof codes \cite{Cohen200005}, etc.

In general, many problems about sequences can be formulated as follows:

\begin{definition}[cf. \cite{Korner199500}]
Suppose $k \ge 2$. For a $k$-tuple $(\seq{v}_i)_{i=1}^k$ of vectors in
$\field{q}^n$, we define the set 
\[
W((\seq{v}_i)_{i=1}^k) \eqdef \{(v_{i,j})_{i=1}^k: j=1,2,\ldots, n\},
\]
where $v_{i,j}$ denotes the $j$-th component of
$\seq{v}_i$.\footnote{Thus
$W((\seq{v}_i)_{i=1}^k)\subseteq\field{q}^k$ is the set of columns of
the matrix with rows $\seq{v}_1,\dots,\seq{v}_k$ (in that order).}
Let $\frak{F}$ be a family of subsets of $\field{q}^k$.
A set $C \subseteq \field{q}^n$ is called an $\frak{F}$-set if
\[
W((\seq{v}_i)_{i=1}^k) \cap S \ne \varnothing \qquad \text{for any
  $k$ distinct $\seq{v}_1, \seq{v}_2, \ldots,
  \seq{v}_k \in C$}
\]
and every $S \in \frak{F}$. The maximum number of elements in an
$\frak{F}$-subset of $\field{q}^n$ is denoted by $N(\frak{F},n)$. 
\end{definition}

For example, a \emph{$(2,1)$-separating system} (see
e.g. \cite{Korner199500}) is an $\frak{F}$-subset of $\field{2}^n$
with
\[
\frak{F} = \{\{(0,0,1), (1,1,0)\}, \{(0,1,0), (1,0,1)\}, \{(1,0,0), (0,1,1)\}\},
\]
and a \emph{$k$-independent family} (see e.g. \cite{Cohen199411, Kleitman197300})
is an $\frak{F}$-subset of $\field{q}^n$ with $\frak{F} =
\{\{\seq{v}\}: \seq{v} \in \field{q}^k\}$. A general asymptotic lower
bound for $N(\frak{F},n)$ is given by the next proposition, which is a
simple extension of the idea in \cite{Korner199500}.

\begin{proposition}\label{pr:GeneralRandomCodingBound}
  Let $\rset{C}(M_n, k) = (\rseq{v}_i)_{i=1}^{M_n}$ be a sequence
  of $M_n$ random vectors in $\field{q}^n$ such that each $\rseq{v}_i$
  is uniformly distributed over $\field{q}^n$ and any $k$ random
  vectors $\rseq{v}_{i_1},\rseq{v}_{i_2},\dots,\rseq{v}_{i_k}$,
  $1\leq i_1<i_2<\dots<i_k\leq M_n$, are independent. Then for any
  given family $\frak{F}$ of subsets of $\field{q}^k$, if
\begin{equation}\label{eq:GeneralRandomCodingBound}
  \limsup_{n\to\infty} \frac{1}{n} \log_q M_n <
  \min\left\{\frac{1}{k-1} 
  \left(k - \log_q \left(q^k - \min_{S \in \frak{F}} |S|\right) \right), 1\right\},
\end{equation}
we can extract from $\rset{C}(M_n, k)$ a random $\frak{F}$-set of size
$\tilde{M}_n$ such that
\[
\lim_{n\to\infty} \frac{\tilde{M}_n}{M_n} = 1 \qquad \text{almost surely}.
\]
\end{proposition}

\begin{proof}
We call a $k$-tuple $\seq{i}=(i_1,\dots,i_k)\in\{1, 2,\dots,M_n\}^k$
``undesirable'' if $i_1$, $i_2$, \ldots, $i_k$ are distinct and
$W\bigl((\rseq{v}_{i_j})_{j=1}^k\bigr) \cap S = \varnothing$ for some
$S\in \frak{F}$, and denote by $\rset{U}$ the random set consisting of all
undesirable $k$-tuples.
For $\seq{i}\in\{1, 2,\dots,M_n\}^k$ with distinct components
$i_1,i_2,\dots,i_k$ and 
$S\in\frak{F}$ we have $P\{W((\rseq{v}_{i_j})_{j=1}^k) \cap S =
  \varnothing\}=(1-|S|/q^k)^n$, since
$(\rseq{v}_{i_1},\rseq{v}_{i_2},\dots,\rseq{v}_{i_k})$ is uniformly
distributed over $(\field{q}^n)^k$. Hence the expected number of
undesirable $k$-tuples satisfies
 
\[
E[|\rset{U}|] \le M_n^k \sum_{S\in\frak{F}} \left(1-\frac{|S|}{q^k}\right)^n
\le |\frak{F}| M_n^k \left(1-\frac{\min_{S\in\frak{F}} |S|}{q^k}\right)^n.
\]
Combined with condition~\eqref{eq:GeneralRandomCodingBound}, this gives
\[
E[|\rset{U}|] \le \frac{\alpha M_n}{2kn^2} \qquad \text{for any constant $\alpha \in (0,1)$ and all $n \ge n_0(\alpha, k)$},
\]
so that $P\{|\rset{U}| \ge \alpha M_n/(2k)\} \le n^{-2}$ for all $n
\ge n_0(\alpha, k)$. 

Next we apply the analogous reasoning with $k'=2$ and the set family 
$\frak{F}'$ having $\bigl\{(u,v)\in\field{q}^2;u\neq v\bigr\}$ as
its single member (a family of subsets of $\field{q}^2$).
Here the random set of ``undesirable pairs'' is
\[
\rset{U}' \eqdef \{(i_1, i_2) \in \{1, 2, \ldots, M_n\}^2; i_1 \ne i_2,
\rseq{v}_{i_1} = \rseq{v}_{i_2}\}.
\]
Using the second inequality $\limsup_{n\to\infty} n^{-1} \log_q
M_n <1$ from \eqref{eq:GeneralRandomCodingBound}, we obtain similarly
$P\{|\rset{U}'| \ge \alpha M_n/4\} \le n^{-2}$ for all $n \ge
n_1(\alpha)$.

Now we remove from $\rset{C}(M_n,k)$ all components which appear in at
least one undesirable $k$-tuple or pair. The resulting sequence
has distinct components and constitutes a random $\frak{F}$-set of 
cardinality $\tilde{M}_n\geq M_n-k|\tilde{U}|-2|\tilde{U}'|$. This gives
\[
P\left\{\frac{\tilde{M}_n}{M_n} \le 1 - \alpha\right\} \le
\frac{2}{n^2} \qquad \text{for all $n \ge \max\{n_0(\alpha, k),
  n_1(\alpha)\}$},
\]
and 
together with the Borel-Cantelli lemma shows that
\[
\liminf_{n\to\infty} \frac{\tilde{M}_n}{M_n} \ge 1 - \alpha \qquad
\text{almost surely}. 
\]
Taking $\alpha = l^{-1}$ with $l = 1,2,\ldots$, we finally obtain
\begin{align*}
P\left\{\lim_{n\to\infty} \frac{\tilde{M}_n}{M_n} = 1\right\}
&= P\left\{\bigcap_{l=1}^\infty \left\{\liminf_{n\to\infty} \frac{\tilde{M}_n}{M_n} \ge 1-\frac{1}{l}\right\}\right\} \\
&= 1 - P\left\{\bigcup_{l=1}^\infty \left\{\liminf_{n\to\infty} \frac{\tilde{M}_n}{M_n} < 1-\frac{1}{l}\right\}\right\} \\
&= 1 - \lim_{l\to\infty} P\left\{\liminf_{n\to\infty} \frac{\tilde{M}_n}{M_n} < 1-\frac{1}{l}\right\} \\
&= 1,
\end{align*}
as desired.
\end{proof}

Proposition~\ref{pr:GeneralRandomCodingBound} tells us that the
asymptotic lower bound \eqref{eq:GeneralRandomCodingBound} for
$N(\frak{F},n)$ can be achieved by a family of special sequences of
random vectors, which may be called \emph{$k$-independent sequences
  (of random vectors)}. The next two propositions provide some ways
for generating a $k$- or $(k+1)$-independent sequence based on a
$k$-good random matrix.

\begin{proposition}\label{pr:Application3A}
  Let $\rmat{A}$ be a $k$-good random $m\times n$ matrix and
  $U=\{\seq{u}_i\}_{i=1}^{M}$ a set of $M$ vectors in $\field{q}^m$
  such that any $k$ of them are linearly independent. Then the random
  mapping $F(i): \{1,2,\ldots,M\} \to \field{q}^n$ given by $i \mapsto
  \seq{u}_i \rmat{A}$ satisfies
\[
P\{F(i_1) = \seq{v}_1, F(i_2) = \seq{v}_2, \ldots, F(i_k) =
\seq{v}_k\} = \prod_{j=1}^k P\{F(i_j) = \seq{v}_j\} = q^{-kn},
\]
where $1\le i_1 < i_2 < \cdots < i_k \le M$, $\seq{v}_1, \seq{v}_2,
\ldots, \seq{v}_k \in \field{q}^n$.
\end{proposition}
\begin{proof}
  By assumption the $k\times m$ matrix with rows
  $\seq{u}_{i_1},\dots,\seq{u}_{i_k}$ has full rank, and the result
  follows immediately from Def.~\ref{df:KGoodRandomMatrix}.
\end{proof}
Finding a set of $M$ vectors in $\field{q}^m$ such that any $k$ of
them are linearly independent is equivalent to finding an $m\times M$
parity-check matrix of a 
$q$-ary linear $[M,K,d]$ code with $K\geq M-m$ and $d\geq k+1$.
It is thus an instance of the packing problem of algebraic coding theory;
see \cite{hirschfeld-storme98,hirschfeld-storme01,ivan-storme11}.\footnote{In
  \cite{hirschfeld-storme98,hirschfeld-storme01} the largest possible size 
  of $U$ (equivalently, the largest number of points in the projective space
  $\PG(m-1,\field{q})$ having the property that any $k$ of them are in
  general position) is denoted by $M_k(m-1,q)$.}
Also note that Prop.~\ref{pr:Application3A} includes
Prop.~\ref{pr:Application1} as a special case.

\begin{proposition}\label{pr:Application3B}
  Let $\rmat{A}$ be a $k$-good random $m\times n$ matrix and
  $\rseq{v}$ a random vector independent of $\rmat{A}$ and
  uniformly distributed over 
  $\field{q}^n$. Let $U=\{\seq{u}_i\}_{i=1}^{M}$ be a set of $M$
  vectors in $\field{q}^m$ such that any $k+1$ of them, as points of
  the affine space $\AG(m,\field{q})$, do not lie on any $(k-1)$-flat of
  $\AG(m,\field{q})$.
  Then the random
  mapping $F(i): \{1,2,\ldots,M\} \to \field{q}^n$ given by $i \mapsto
  \seq{u}_i \rmat{A} + \rseq{v}$ satisfies
\[
P\{F(i_1) = \seq{v}_1, F(i_2) = \seq{v}_2, \ldots, F(i_{k+1}) =
\seq{v}_{k+1}\} = \prod_{j=1}^{k+1} P\{F(i_j) = \seq{v}_j\} =
q^{-(k+1)n},
\]
where $1\le i_1 < i_2 < \cdots < i_{k+1} \le M$, $\seq{v}_1,
\seq{v}_2, \ldots, \seq{v}_{k+1} \in \field{q}^n$.
\end{proposition}
\begin{proof}
  This can be seen as an ``affine analogue'' of
  Prop.~\ref{pr:Application3A}, and is proved as
  follows:\footnote{Note that we cannot directly apply
    Prop.~\ref{pr:Application3A}, since the random
    $(m+1)\times n$ matrix
    formed from $\rmat{A}$ and $\rseq{v}$ need not be $(k+1)$-good.}
\begin{align*}
&P\{F(i_j) = \seq{v}_j; j = 1, \ldots, k+1\} \\
&\quad = \sum_{\substack{\mat{A} \in \supp{\rmat{A}}\\ \seq{v} \in \field{q}^n}}
q^{-n} P\{\rmat{A} = \mat{A}\} 1\{\seq{u}_{i_j} \mat{A} + \seq{v} = \seq{v}_j; j = 1, \ldots, k+1\} \\
&\quad = q^{-n} \sum_{\mat{A} \in \supp{\rmat{A}}}
\Big( P\{\rmat{A} = \mat{A}\} 1\{(\seq{u}_{i_j}-\seq{u}_{i_1}) \mat{A} = \seq{v}_j - \seq{v}_1; j = 2, \ldots, k+1\} \\
&\qquad \times \sum_{\seq{v} \in \field{q}^n} 1\{\seq{v} = \seq{v}_1 - \seq{u}_{i_1}\mat{A}\} \Big) \\
&\quad = q^{-n} P\{(\seq{u}_{i_j}-\seq{u}_{i_1}) \rmat{A} = \seq{v}_j - \seq{v}_1; j = 2, \ldots, k+1\} \\
&\quad = q^{-(k+1)n},
\end{align*}
where the last equality follows from the $k$-goodness of $\rmat{A}$
and the linear independence of
$\{\seq{u}_{i_j}-\seq{u}_{i_1}\}_{j=2}^{k+1}$.
\end{proof}

Proposition~\ref{pr:Application3B} provides an alternative way for
generating a $k$-independent sequence. For example, to generate a
$2$-independent sequence, we may use a $1$-good random matrix and
simply choose $U=\field{q}^{m}$, which has been well known in the
random coding approach (see, e.g., \cite{Yang200904} and the
references therein). Similarly, for a $3$-independent sequence over
$\field{2}$, we may use a $2$-good random matrix over $\field{2}$ and
let $U=\field{2}^m$. The next corollary states this fact.

\begin{corollary}
  Let $\rmat{A}$ be a $2$-good random matrix over $\field{2}$ and
  $\rseq{v}$ a random vector independent of $\rmat{A}$ and uniformly
  distributed over $\field{2}^n$. The random mapping $F: \field{2}^m
  \to \field{2}^n$ given by $\seq{u} \mapsto \seq{u}\rmat{A} +
  \rseq{v}$ satisfies
\[
P\{F(\seq{u}_1) = \seq{v}_1, F(\seq{u}_2) = \seq{v}_2, F(\seq{u}_3) =
\seq{v}_3\} = \prod_{i=1}^3 P\{F(\seq{u}_i) = \seq{v}_i\} = 2^{-3n}, 
\]
for distinct $\seq{u}_1, \seq{u}_2, \seq{u}_3 \in \field{2}^m$ and
arbitrary $\seq{v}_1, \seq{v}_2, \seq{v}_3 \in \field{2}^n$.
\end{corollary}

Similar to Prop.~\ref{pr:Application1}, the corollary depends in
an essential way on the fact that no three points in an affine space
of order $2$ are collinear.

\section{Dense Sets of Matrices}
\label{sec:AlgebraicAspects}


In this section we investigate a fundamental property of sets
of $m\times n$ matrices over $\field{q}$, which is shared by
the support of every $k$-good random $m\times n$ matrix over $\field{q}$.
Definition~\ref{df:KGoodRandomMatrix} implies that for any
full-rank matrix $\mat{M} \in \field{q}^{k\times m}$ and any matrix
$\mat{K} \in \field{q}^{k\times n}$, there exists a matrix $\mat{A}
\in \supp{\rmat{A}}$ such that $\mat{M} \mat{A} = \seq{K}$. 
This motivates the following
\begin{definition}
  \label{df:k-dense}
  Let $k,m,n$ be positive integers with $k \le \min\{m,n\}$. A set
  $\matset{A} \subseteq \field{q}^{m \times n}$ is said to be
  \emph{$k$-dense} if $\mat{M} \matset{A} = \field{q}^{k\times n}$ for
  every full-rank matrix $\mat{M} \in \field{q}^{k\times m}$. As in
  the case of ``good'' we use the terms \emph{$1$-dense} and
  \emph{dense} interchangeably.
\end{definition}

In the sequel, we tacitly assume that $k \le \min\{m,n\}$. In the
language of linear transformations Def.~\ref{df:k-dense}
requires that for every $k$-dimensional subspace
$U\subseteq\field{q}^m$ the restriction map
$\matset{A}\to\Hom(U,\field{q}^n)$, $f\mapsto f|_U$ is surjective. The
definition also has a nice geometric interpretation, which we now
proceed to discuss.

\begin{definition}
  \label{df:RAG(m,n,q)} The left affine space of
  $m\times n$ matrices over $\field{q}$, denoted by
  $\LAG(m,n,\field{q})$, is the lattice of cosets (including the empty set)
  of left $R_m$-submodules of
  $\field{q}^{m\times n}$. A coset $\mat{A}+\matset{U}$ is called an
  $r$-dimensional flat ($r$-flat) if
  $\matset{U}\cong\field{q}^{m\times r}$ as an
  $R_m$-module. 
Equivalently, the subspace of $\field{q}^n$
generated by the row spaces of all matrices in $\matset{U}$ has
dimension $r$; cf.\ the remarks in
Section~\ref{sec:whom}.\footnote{Note also that the $R_m$-modules
  $\field{q}^{m\times r}$, $r=1,2,\dots$, form a set of
  representatives for the isomorphism classes of finitely
  generated $R_m$-modules (a special case of a general 
  theorem about simple Artinian rings \cite[Ch.~1]{lam91}).}
\end{definition}
As usual, flats of dimension $0$, $1$, $2$, $n-1$ are called points,
lines, planes, and hyperplanes, respectively. The whole geometry
(i.e.\ the flat $\field{q}^{m\times n}$) has dimension
$n$. 

Right affine spaces of rectangular matrices are defined in an
analogous manner and denoted by $\RAG(m,n,\field{q})$. Since
$\mat{A}\mapsto\mat{A}^T$ defines an isomorphism between $\LAG(m,n,\field{q})$
and $\RAG(n,m,q)$, it is sufficient to consider only left (or only right)
affine spaces of rectangular matrices. For notational convenience we
have started our discussion with left affine spaces. Right affine
spaces, being the appropriate framework for dense sets as defined
above, will be used after the following remark.

\begin{remark}
  \label{re:RAG(m,n,q)}
  Other descriptions of $\LAG(m,n,\field{q})$ have appeared in the
  literature. Consider the geometry $H_q^{(m+n-1,m-1)}$ whose points
  are the $(m-1)$-flats of $\PG(m+n-1,q)$ skew to a fixed $(n-1)$-flat
  $W$, and whose lines are the $m$-flats of $\PG(m+n-1,q)$ meeting $W$
  in a point. The map which sends $\mat{A}\in\field{q}^{m\times n}$ to
  the row space of $(\mat{I}_m,\mat{A})$ is easily seen to
  define an isomorphism from $\LAG(m,n,\field{q})$ onto
  $H_q^{(m+n-1,m-1)}$. The geometry $\LAG(m,n,\field{q})$ is related to the
  space of rectangular $m\times n$-matrices over $\field{q}$ (see
  \cite[Ch.~3]{wan96a}), 
  but it is not the same.\footnote{In the Geometry of Matrices lines are
    defined as cosets of $\field{q}$-subspaces generated by rank-one
    matrices.}
  The special case $m=2$, in which
  $H_q^{(n+1)*}\eqdef H_q^{(n+1,1)}$ is an example of a so-called semipartial
  geometry, is discussed in \cite[2.2.7]{clerck-maldeghem95}.

  Another description is by means of a so-called linear representation
  in the ordinary affine space $\AG(n,\field{q^m})$ over the extension
  field $\field{q^m}$: Identify the point set of $\LAG(m,n,\field{q})$
  with that of $\AG(n,\field{q^m})$ by viewing the columns of
  $\mat{A}\in\field{q}^{m\times n}$ as coordinate vectors with respect
  to a fixed basis of $\field{q^m}$ over $\field{q}$. It is then
  readily verified that the lines of $\LAG(m,n,\field{q})$ correspond
  exactly to those lines of $\AG(n,\field{q^m})$, whose associated
  $1$-dimensional subspace (``direction'') is spanned by a vector in
  $\field{q}^n$. In other words, a line $L$ of $\AG(n,\field{q^m})$
  belongs to $\LAG(m,n,\field{q})$ if and only if $L$ meets the
  hyperplane $\PG(n-1,\field{q^m})$ at infinity in a point of the subgeometry
  $\PG(n-1,\field{q})$.\footnote{Thus the line set of
    $\LAG(m,n,\field{q})$ is the union of $\frac{q^n-1}{q-1}$ parallel
    classes of lines of $\AG(n,\field{q^m})$.}  Again the special case
  $m=2$ is mentioned in \cite[2.3.2]{clerck-maldeghem95}.
\end{remark}
The following lemma provides the geometric interpretation of $k$-dense
sets of matrices and the link with $k$-good random matrices.
\begin{lemma}
  \label{le:blocking}
  Let $\matset{A}$ be a nonempty subset of $\field{q}^{m\times n}$ and
  $\rmat{A}$ the random $m\times n$ matrix uniformly distributed over
  $\matset{A}$.
  \begin{enumerate}[(i)]
  \item $\matset{A}$ is $k$-dense if and
  only if it meets every $(m-k)$-flat of $\RAG(m,n,\field{q})$ in at
  least one point, i.e.,
  $\matset{A}$ is a blocking set with respect to $(m-k)$-flats in
  $\RAG(m,n,\field{q})$.
\item $\rmat{A}$ is $k$-good if and
  only if $\matset{A}$ meets every $(m-k)$-flat of $\RAG(m,n,\field{q})$
  in the same number, say $\lambda$, of points.
\end{enumerate}
\end{lemma}
The condition in (ii) is equivalent to $\matset{A}$ being a $k$-design of index
$\lambda$ in the sense of \cite{Delsarte197800}.\footnote{For this one
  has to identify matrices in $\field{q}^{m\times n}$ with bilinear
  forms $\field{q}^m\times\field{q}^n\to\field{q}$.}
Since $\field{q}^{m\times n}$ is the union of $q^{kn}$ flats parallel
to a given $(m-k)$-flat, the constant $\lambda$ in (ii) must be equal to
$|\matset{A}|q^{-kn}$.
\begin{proof}
  Obviously a necessary and sufficient condition for
  $\matset{A}\subseteq\field{q}^{m\times n}$ to be $k$-dense is that
  $\matset{A}$ contains a set of coset representatives for every
  annihilator subspace
  $\matset{U}=\mat{M}^\perp\eqdef\{\mat{X}\in\field{q}^{m\times
    n};\mat{MX}=\mat{0}\}$ with $\mat{M}\in\field{q}^{k\times m}$ of
  full rank. We have $\mat{A}\in\matset{U}$ if and only if the column
  space of $\mat{X}$ is contained in the orthogonal complement (an
  $(m-k)$-dimensional space) of the row space of $\mat{M}$. Hence these
  annihilator subspaces are exactly the $(m-k)$-flats of
  $\RAG(m,n,\field{q})$ through $\mat{0}$, and Part~(i) is proved.
  
  For (ii) take an arbitrary $\mat{K}\in\field{q}^{k\times n}$ and 
  note that $P\{\mat{M}\rmat{A} = \mat{K}\}
  =|\{\mat{A}\in\matset{A};\mat{MA}=\mat{K}\}|/|\matset{A}|
  =|\matset{A}\cap(\matset{U}+\mat{A}_0)|/|\matset{A}|$, where
  $\mat{A}_0\in\field{q}^{m\times n}$ is any matrix with $\mat{MA}_0=\mat{K}$.
\end{proof}
As in the case of $k$-good matrices, we are interested in the minimum
size of a $k$-dense subset of $\field{q}^{m\times n}$. We denote this
size by $\minsize_k(m,n,\field{q})$. With the aid of our previous results it
will be easy to determine the numbers $\minsize_k(m,n,\field{q})$ for $m\leq
n$. The case $m>n$ is considerably more difficult and includes, for
example, the problem of determining the minimum size of a blocking set
with respect to $(m-k)$-flats in the ordinary affine space
$\AG(m,\field{q})\cong\RAG(m,1,\field{q})$, which (despite a lot of
research by finite geometers) remains unsolved in general. Here we will restrict
ourselves to easy-to-derive bounds and then work out completely the
smallest nontrivial case.


We start with the case $m\leq n$.
\begin{theorem}
  \label{thm:mleqn}
  For $k\leq m\leq n$ we have $\minsize_k(m,n,\field{q})=q^{kn}$, and
  a subset $\matset{A}\subseteq\field{q}^{m\times n}$ of size $q^{kn}$
  is $k$-dense if and only if it is a (not necessarily linear)
  $(m,n,k)$ MRD code.
\end{theorem}
\begin{proof}
  Since $m\leq n$, an $(m,n,k)$ MRD code has size $q^{kn}$ and meets
  every $(m-k)$-flat of $\RAG(m,n,\field{q})$ in exactly $1$ point (cf.\
  Lemma~\ref{le:MRDCode}). On the other hand a $k$-dense subset of
  $\field{q}^{m\times n}$ must have size at least $q^{kn}$ (consider a
  partition of $\field{q}^{m\times n}$ into $q^{kn}$ parallel $(m-k)$-flats).
  \end{proof}
The (nontrivial) minimum blocking sets of Th.~\ref{thm:mleqn} have no analogue
in the geometries $\PG(m,\field{q})$ or $\AG(m,\field{q})$: For $1\leq
k\leq m$ the only
point sets in $\PG(m,\field{q})$ or $\AG(m,\field{q})$ meeting every
$(m-k)$-flat in the same number of points are $\emptyset$ and the whole
point set. This is due to the fact that the corresponding incidence
structures are nontrivial $2$-designs and hence their incidence
matrices have full rank; cf.\ \cite[p.~20]{dembowski68} or the proof
of Fisher's Inequality in \cite[Ch.~10.2]{mhall86}.

If $m\leq n$ then
in $\RAG(m,n,\field{q})$ there exist, for every $t\in\{1,2,\dots,q^{(m-k)n}\}$,
point sets meeting every $(m-k)$-flat in $t$ points.
For example, the union of any $t$ cosets of a fixed linear $(m,n,k)$ MRD
code has this property. The property of meeting all flats of a fixed
dimension in the same number of points is in fact left-right symmetric:
\begin{theorem}
  \label{th:leftright}
  If $\matset{A}\subseteq\field{q}^{m\times n}$ meets every $(m-k)$-flat
  of $\RAG(m,n,\field{q})$ in the same number, say $\lambda$, of
  points, then the same is true for the $(n-k)$-flats of $\LAG(m,n,\field{q})$
  (the corresponding number being $\lambda'=\lambda q^{k(n-m)})$. 
\end{theorem}
\begin{proof}
  This follows from Lemma~\ref{le:blocking}(ii) and
  Th.~\ref{th:KGoodRandomMatrixDuality}.
\end{proof}
Th.~\ref{th:leftright} does not require that $m\leq n$. In the
case $m>n$ it says, mutatis mutandis, that subsets
$\matset{A}\subseteq\field{q}^{m\times n}$ meeting every $(m-k)$-flat
of $\RAG(m,n,\field{q})$ in a constant number $\lambda$ of points
exist only if $\lambda$ is a multiple of $q^{k(m-n)}$, the smallest
such sets being again the $(m,n,k)$ MRD codes (the case
$\lambda=q^{k(m-n)}$).

\begin{remark}
  \label{re:completesubspace}
  If $\matset{A}\subseteq\field{q}^{m\times n}$ is both $k$-dense and
  an $\field{q}$-subspace of $\field{q}^{m\times n}$, then
  $\matset{A}$ meets every $(m-k)$-flat of $\RAG(m,n,\field{q})$ in the same
  number of points and gives rise to a $k$-good random $m\times
  n$-matrix $\rmat{A}$ in the sense of
  Lemma~\ref{le:blocking}. (This follows by consideration of the
  $\field{q}$-linear maps $\matset{A}\to\field{q}^{k\times n}$,
  $\mat{A}\mapsto\mat{MA}$, with $\mat{M}\in\field{q}^{k\times n}$
  of full rank, which are surjective.) Moreover, Th.~\ref{th:leftright} (or
  Th.~\ref{th:KGoodRandomMatrixDuality}) applies, showing that the
  property ``$k$-dense $\field{q}$-subspace'' is preserved under
  $\mat{A}\mapsto\mat{A}^T$ and the minimum dimension of a $k$-dense
  $\field{q}$-subspace of $\field{q}^{m\times n}$ is $k\max\{m,n\}$.  
\end{remark}

From now on we assume $m>n$. First we collect some general information
about the numbers $\minsize_k(m,n,\field{q})$ in this case.
 

\begin{theorem}
  \label{th:m>n}
  \begin{enumerate}[(i)]
  \item $\minsize_1(m,1,q)=1+m(q-1)$ for all $m\geq 2$.
  \item For $1\leq k\leq n<m$ we have the bounds
    $q^{kn}<\minsize_k(m,n,\field{q})<q^{km}$.
  \end{enumerate}
\end{theorem}
\begin{proof}
  (i) A dense subset of $\field{q}^{m\times 1}$ is the same as a
  blocking set with respect to hyperplanes in the ordinary affine
  space $\AG(m,\field{q})$. 
  The minimum size of such a blocking set is known
  to be $1+m(q-1)$ and is realized (among other configurations) by the
  union of $m$ independent lines through a fixed point of
  $\AG(m,\field{q})$; see the original sources
  \cite{brouwer-schrijver78,jamison77}
  or \cite[Cor.~2.3]{ball11}, \cite[Th.~6.1]{blokhuis-sziklai-szonyi11}.

  (ii) For the lower bound note that a set
  $\matset{A}\subseteq\field{q}^{m\times n}$ of size
  $|\matset{A}|=q^{kn}$ has $\rankd(\matset{A})\leq m-k$. (By the
  Singleton bound $kn\leq m(n-\rankd(\matset{A})+1)$, so that
  $\rankd(\matset{A})\leq\frac{(m-k)n}{m}+1<m-k+1$.) Hence there exist
  $\mat{A}_1,\mat{A}_2\in\matset{A}$ such that the column space
  $U\subset\field{q}^m$ of
  $\mat{A}_1-\mat{A}_2$ has dimension $\leq m-k$. For
  a full-rank matrix $\mat{B}\in\field{q}^{k\times m}$ with row space
  contained in $U^\perp$ (such $\mat{B}$ exists since $\dim U^\perp\geq k$)
  we then have $\mat{BA}_1=\mat{BA}_2$, so that
  $|\{\mat{BA};\mat{A}\in\matset{A}\}|<q^{kn}=|\field{q}^{k\times
    n}|$. This shows that $\matset{A}$ is not dense, i.e.\
  $\minsize_k(m,n,\field{q})>q^{kn}$.

  For the upper bound in (ii) choose $\matset{A}$ as a linear
  $(m,n,k)$ MRD code. Then $|\matset{A}|=q^{km}$, and by
  Th.~\ref{th:leftright} the set $\matset{A}$ meets
  every $(m-k)$-flat of $\RAG(m,n,\field{q})$ in $q^{k(m-n)}$ points.
  Hence any subset $\matset{A}'\subseteq\matset{A}$ of
  size at least $q^{km}-q^{k(m-n)}+1$ is still $k$-dense, and consequently
  $\minsize_k(m,n,\field{q})\leq q^{km}-q^{k(m-n)}+1<q^{km}$.
\end{proof}

The bounds in Th.~\ref{th:m>n}(ii) are
rather weak and serve only to refute the obvious guesses
``$\minsize_k(m,n,\field{q})=q^{kn}$'' or
``$\minsize_k(m,n,\field{q})=q^{km}$''. Refining these bounds will be left
for subsequent work. Instead we will now work out completely the
binary case $(m,n)=(3,2)$ (the smallest case left open by
Th.~\ref{th:m>n}).


Before giving the result we will collect a few combinatorial facts
about the geometry $\RAG(3,2,\field{2})$. Perhaps the most important
property is that the substructure consisting of all 
  lines and planes through a fixed point is isomorphic to
  $\PG(2,\field{2})$. (This can be seen from the parametrization of
  the flats through $\mat{0}$, the all-zero matrix,
  by subspaces of $\field{2}^3$.) 
  In particular each point is contained in $7$ lines and $7$ planes.

  Altogether there are $64$ points, $112$ lines
  falling into $7$ parallel classes of size $16$, and $28$ planes
  falling into $7$ parallel classes of size $4$. A line contains
  $4$ points and is contained in $3$ planes. A plane contains $16$
  points and $12$ lines ($3$ parallel classes of size $4$). Two
  distinct points $\mat{A}_1$, $\mat{A}_2$ are incident with a unique
  line (and hence with $3$ planes) if $\rank(\mat{A}_1-\mat{A}_2)=1$, and
  incident with a unique plane (but not with a line)
  if $\rank(\mat{A}_1-\mat{A}_2)=2$.

  Let us now consider a plane of $\RAG(3,2,\field{2})$, which is
  isomorphic to $\RAG(2,2,\field{2})$. In $\field{2}^{2\times 2}$ there are
  $9$ matrices of rank $1$ (accounting for the nonzero points on the
  $3$ lines through $\mat{0}$) and $6$ matrices of rank $2$, which
  together with $\mat{0}$ form two (linear) MRD codes:
  \begin{align*}
    \matset{A}&=\left\{
        \begin{pmatrix}
          0&0\\0&0
        \end{pmatrix},
        \begin{pmatrix}
          1&0\\0&1
        \end{pmatrix},
        \begin{pmatrix}
          1&1\\1&0
        \end{pmatrix},
        \begin{pmatrix}
          0&1\\1&1
        \end{pmatrix}
        \right\},\\
        \matset{A}'&=\left\{
        \begin{pmatrix}
          0&0\\0&0
        \end{pmatrix},
        \begin{pmatrix}
          0&1\\1&0
        \end{pmatrix},
        \begin{pmatrix}
          1&0\\1&1
        \end{pmatrix},
        \begin{pmatrix}
          1&1\\0&1
        \end{pmatrix}
        \right\}.
  \end{align*}
  Every further MRD code is a coset of either $\matset{A}$ or
  $\matset{A}'$. The $12$ lines of $\RAG(2,2,\field{2})$ and the $8$
  MRD codes impose on $\field{2}^{2\times 2}$ the structure of the
  affine plane of order $4$.\footnote{This is particularly visible in
    the second model of $\RAG(2,2,\field{2})$ described in
    Rem.~\ref{re:RAG(m,n,q)}. The sets $\matset{A}$, $\matset{A}'$ are
    the two lines of $\AG(2,\field{4})$ of the form
    $\field{4}(1,\alpha)$ with
    $\alpha\in\field{4}\setminus\field{2}$.}

Clearly the MRD codes are exactly the blocking sets (with respect to lines)
of $\RAG(2,2,\field{2})$ of minimum size $4$. For slightly larger
blocking sets we have the following result, which is needed in the proof
of our next theorem. 
  \begin{lemma}
    \label{le:(3,2)}
    Let $\matset{S}$ be a blocking set in the plane $\RAG(2,2,\field{2})$.
    \begin{enumerate}[(i)]
    \item If $|\matset{S}|=5$ then $\matset{S}$ contains an MRD
      code and is disjoint from another MRD code.
    \item If $|\matset{S}|=6$ then $\matset{S}$ is disjoint from
      an MRD code.
    \item If $|\matset{S}|\in\{7,8\}$ then there exist $3$ mutually
      non-collinear points outside $\matset{S}$.
      \end{enumerate}
  \end{lemma}
  \begin{proof}
    (i) Consider $\matset{S}$ as a set of points in
    $\AG(2,\field{4})$. If $\matset{S}$ meets every line
    of $\AG(2,\field{4})$ in at most
    $2$ points, then $\matset{S}$ is an oval and has
    intersection pattern $2,2,1,0$ with each parallel class of
    lines of $\AG(2,\field{4})$. In particular $\matset{S}$ is not a
    blocking set in $\RAG(2,2,\field{2})$. Likewise, if $\matset{S}$
    meets some line of $\RAG(2,2,\field{2})$ in at least $3$
    points, it cannot be a blocking set (since it cannot block all lines 
    in the corresponding parallel class). 
    Hence $\matset{S}$ meets some line $\matset{L}$ outside
    $\RAG(2,2,\field{2})$ (i.e.\ an MRD code)
    in at least $3$ points. But then $\matset{S}$ must also contain the
    $4$-th point on $\matset{L}$, which is the point of concurrency of
    $3$ lines in $\RAG(2,2,\field{2})$. Thus $\matset{S}$ contains
    $\matset{L}$ and is disjoint from some line (MRD code) parallel to
    $\matset{L}$.

    (ii) This follows from the fact that the minimum size of a
    blocking set in $\AG(2,\field{4})$ is $7$.
    
    (iii) Since the point set can be partitioned into $4$ MRD codes
    (in two different ways), this is clear for $|\matset{S}|=7$, and
    for $|\matset{S}|=8$ it can fail only if $\matset{S}$ meets every
    MRD code in $2$ points. However a set $\matset{S}$ with this property 
    must contain a line of $\AG(2,\field{4})$
    and, by symmetry, also be disjoint from some line. Hence it cannot
    be a blocking set in $\RAG(2,2,\field{2})$.\footnote{If there is
      no $4$-line, then every point of $\matset{S}$ is on three
      $2$-lines and two $3$-lines, which gives $8\cdot 2/3$ as the
      number of $3$-lines, a contradiction.}  
  \end{proof}
  The following property of $\RAG(3,2,\field{2})$ will also be needed.
  \begin{lemma}
    \label{le:triangle}
    If $x$, $y$ are two non-collinear points of $\RAG(3,2,\field{2})$,
    then every point $z$ collinear with both $x$ and $y$ is contained
    in the (unique) plane generated by $x$ and $y$.
  \end{lemma}
  \begin{proof}
    This is evident, since the plane spanned by the lines $xz$ and
    $yz$ contains $x$, $y$, and $z$.
  \end{proof}
  Now we are ready to state and prove our theorem concerning the numbers
  $\minsize_k(3,2,\field{2})$.
\begin{theorem}
  \label{th:(3,2)}
  \begin{enumerate}[(i)]
  \item $\minsize_1(3,2,\field{2})=6$;
    \item $\minsize_2(3,2,\field{2})=22$.
  \end{enumerate}
\end{theorem}
\begin{proof}
  

  (i) Let $\matset{A}$ be a
  dense set of points, i.e.\ $\matset{A}$ meets every plane of
  $\AG(3,2,\field{2})$. In Th.~\ref{th:m>n}(ii) we have
  seen that $|\matset{A}|\geq 5$. If $|\matset{A}|=5$ then exactly one
  plane in every parallel class contains $2$ points of $\matset{A}$
  and the three other planes contain exactly $1$ point of
  $\matset{A}$. Thus $7$ planes meet $\matset{A}$ in $2$ points and
  $21$ planes meet $\matset{A}$ in $1$ point. On the other hand, every
  pair of distinct points of $\matset{A}$ is contained in at least one
  plane. Since there are $10$ such pairs, we have a
  contradiction. Hence necessarily $|\matset{A}|\geq 6$. In order to
  finish the proof, we exhibit a dense set $\matset{A}$ with
  $|\matset{A}|=6$:
  \begin{equation}
    \label{eq:dense6}
    \begin{pmatrix}
      1&0\\
      0&1\\
      0&0
    \end{pmatrix},
    \begin{pmatrix}
      1&1\\
      1&0\\
      0&0
    \end{pmatrix},
    \begin{pmatrix}
      0&1\\
      1&1\\
      0&0
    \end{pmatrix},
    \begin{pmatrix}
      0&0\\
      1&0\\
      1&0
    \end{pmatrix},
    \begin{pmatrix}
      0&1\\
      0&0\\
      0&1
    \end{pmatrix},
    \begin{pmatrix}
      0&0\\
      0&0\\
      1&1
    \end{pmatrix}.
  \end{equation}

(ii) First we exhibit a blocking set with respect to lines of size
$22$. We identify $\field{2}^{3\times 2}$ with $\field{8}^2$ by
viewing the columns of $\mat{A}\in\field{2}^{3\times 2}$ as coordinate
vectors relative to $1,\alpha,\alpha^2$, where
$\alpha^3+\alpha+1=0$. The lines of $\AG(2,\field{8})$ fall into three
parallel classes of lines belonging to $\LAG(3,2,\field{2})$ and
represented by $\field{8}(1,0)$, $\field{8}(0,1)$, $\field{8}(1,1)$,
and six parallel classes defining MRD codes in $\field{2}^{3\times 2}$
and represented by $\matset{M}_i=\field{8}(1,\alpha^i)$, $1\leq i\leq 6$. We now
set 
\begin{equation}
  \label{eq:22set}
  \matset{A}=\matset{M}_1\cup \matset{M}_2\cup\matset{M}_4,
\end{equation}
and show that $\matset{A}$ is a blocking set with respect to lines in
$\RAG(3,2,\field{2})$.

Any point $(a,b)\notin\matset{M}_i$, where $i\in\{1,2,4\}$ is fixed,
is collinear in $\RAG(3,2,\field{2})$ with exactly
$3$ points of $\matset{M}_i$ (the points $(x,y)\in\matset{M}_i$
satisfying $x=a$, $y=b$, or $x+y=a+b$). In particular this applies to
the nonzero points on one of the MRD codes $\matset{M}_i$, 
when considered relative to the other two MRD codes.
Now consider a line $\matset{L}$ of $\RAG(3,2,\field{2})$ containing
$2$ points of $\matset{A}$. We claim that $\matset{L}$ contains a
third point of $\matset{A}$, and hence that $\matset{L}$ meets
$\matset{A}$ in exactly $3$ points (since it is transversal to each
$\matset{M}_i$).

Viewed as a subset of $\field{8}^2$, the line
$\matset{L}$ has the form
$\bigl\{(x,y),(x+z,y),(x,y+z),(x+z,y+z)\big\}$ with $z\neq 0$ and
$(x,y)\notin\bigl\{(0,0),(z,0),(0,z),(z,z)\bigr\}$. Replacing $\alpha$
by $\alpha^2$ or $\alpha^4$, if necessary, we may assume
$(a,a\alpha)\in\matset{M}_1\cap\matset{L}$,
$(b,b\alpha^2)\in\matset{M}_2\cap\matset{L}$. Then either $a=b$,
$a\alpha=b\alpha^2$, or $a+a\alpha=b+b\alpha^2$. In the first case
$(a,a\alpha)$, $(a,a\alpha^2)\in\matset{L}$. Since
$1+\alpha+\alpha^2=\alpha^5$, the
remaining two points on $\matset{L}$ are $(a\alpha^5,a\alpha)$,
$(a\alpha^5,a\alpha^2)$, the last point being on $\matset{M}_4$.
In the remaining two cases we conclude similarly
$\matset{L}=\bigl\{(a,a\alpha),(a\alpha^6,a\alpha),(a,a\alpha^4),
(a\alpha^6,a\alpha^4)\bigr\}$, respectively, $\matset{L}=
\bigl\{(a,a\alpha),(a\alpha^4,a\alpha^6),(a,a\alpha^6),
(a\alpha^4,a\alpha)\bigr\}$. In all cases $\matset{L}$ contains a
point of $\matset{M}_4$, completing the proof of our claim.

With this property at hand it is now easy to show that $\matset{A}$ is
a blocking set with respect to lines in $\RAG(3,2,\field{2})$. Indeed,
$\matset{A}$ meets $7$ lines in $1$ point (the lines through
$\mat{0}$), $21$ lines in $3$ points (the lines connecting two points
in different MRD codes $\matset{M}_i$ and $\matset{M}_j$), and hence
$4\cdot 21=84$ lines in $1$ point (the remaining lines through a point
in one of the MRD codes). This accounts for $7+21+84=112$ lines, i.e.\
all the lines of $\RAG(3,2,\field{2})$,
proving the assertion.

Next we show that any blocking set $\matset{A}$ with respect to lines
in $\RAG(3,2,\field{2})$ has size at least $22$. To this end we
consider the plane sections of $\matset{A}$ with respect to some
parallel class $\matset{H}_1,\matset{H}_2,\matset{H}_3,\matset{H}_4$
of planes of $\RAG(3,2,\field{2})$. The indexing will be chosen in such a
way that $i\mapsto|\matset{A}\cap \matset{H}_i|$ is nondecreasing.

Since $\matset{A}\cap \matset{H}_i$ must be a blocking set in
$\matset{H}_i$, we have $|\matset{A}\cap \matset{H}_i|\geq 4$. If
$|\matset{A}\cap \matset{H}_i|\in\{4,5,6\}$ then there exists, by
Lemma~\ref{le:(3,2)}, an MRD code $\matset{M}\subset \matset{H}_i$
disjoint from $\matset{A}$. From the $7\times 4=28$ lines through the
points of $\matset{M}$ exactly $4\times 4=16$ are not contained in
$\matset{H}_i$. By Lemma~\ref{le:triangle} these lines are pairwise
disjoint and hence incident with at least $16$ points of $\matset{A}$.
This shows
\begin{equation*}
  |\matset{A}|\geq|\matset{A}\cap\matset{H}_i|+16,\quad\text{provided
    that $|\matset{A}\cap\matset{H}_i|\in\{4,5,6\}$}.
\end{equation*}
If some section $\matset{A}\cap\matset{H}_i$ has size $6$, we have
$|\matset{A}|\geq 22$ and are done. This leaves as obstructions
the intersection patterns $(4,4,4,8)$, $(4,4,5,8)$, $(4,5,5,7)$ with the
planes $\matset{H}_i$.
However, in such a case there exist, again by
Lemma~\ref{le:(3,2)}, MRD codes
$\matset{M}_i\subseteq\matset{A}\cap\matset{H}_i$ for $i=1,2,3$ and a
set $\matset{M}_4=\{x_1,x_2,x_3\}$ of three pairwise non-collinear points with
$\matset{M}_4\cap\matset{A}=\emptyset$ in the plane $\matset{H}_4$.
Each of the $4$ lines through a point $x_i$ which are
not contained in $\matset{H}_4$ must be incident with at least one point of
$\matset{A}$, and these $4$ (or more) points must be
transversal to
$\matset{M}_i$ for $i=1,2,3$.\footnote{By ``transversal'' we
  mean that $\matset{M}_i$ contains at most one of the points. (It
may contain none of the points.)}
Moreover, the point sets
arising in this way from $x_1,x_2,x_3$ must be pairwise disjoint.
In all three cases this cannot be
accomplished, showing that $|\matset{A}|\leq 21$ is impossible. 
\end{proof}

\section{Conclusion}
\label{sec:Conclusion}

Good random matrices are of particular interest and importance. In
this paper, we determined the minimum support size of a $k$-good
random matrix and established the relation between $k$-good random
matrices with minimum support size and MRD codes. We showed
that homogeneous weights on matrix rings give rise to $k$-good random
matrices, and we explored the connections between $k$-good random
matrices, $k$-dense sets of matrices and the geometry of certain matrix spaces.

However, our understanding of $k$-good random matrices
is still limited. Consider for example the following problem about
the collection $\grmat_k(m,n,\field{q})$ of all $k$-good random $m\times n$
matrices over $\field{q}$, which clearly forms a convex polytope in
$q^{mn}$-dimensional Euclidean space.

\begin{problem}
  \label{pr:polytope}
  Determine the structure of $\grmat_k(m,n,\field{q})$ in more detail.
\end{problem}
A basic question is that about the vertices of this polytope. It is
clear from Th.~\ref{th:KGoodRandomMatrixWithMinimumSupportSize}
that every $(m,n,k)$ MRD code over $\field{q}$ gives rise to a vertex
of $\grmat_k(m,n,\field{q})$ (the uniform distribution over this code). 
But there are other vertices, as the following binary example with
$(m,n)=(2,3)$ shows. 

The geometry $\RAG(2,3,\field{2})$ has $3$ parallel classes of lines
represented by the subspaces (lines through $\mat{0}$)
\begin{equation}
  \label{eq:(2,3)}
  \matset{L}_1=
  \begin{pmatrix}
    a_1&a_2&a_3\\
    0&0&0
  \end{pmatrix},\quad
  \matset{L}_2=
  \begin{pmatrix}
    0&0&0\\
    a_1&a_2&a_3
  \end{pmatrix},\quad
  \matset{L}_3=
  \begin{pmatrix}
     a_1&a_2&a_3\\
     a_1&a_2&a_3
  \end{pmatrix},
\end{equation}
where it is understood that $(a_1,a_2,a_3)$ runs through all of
$\field{2}^3$. Let us choose $3$ nonzero points
$\mat{A}_i\in\matset{L}_i$, $i=1,2,3$, which are linearly dependent
over $\field{2}$ (for example the points obtained by setting $a_1=1$,
$a_2=a_3=0$) and a $4$-dimensional $\field{2}$-subspace $\matset{A}$
of $\field{2}^{2\times 3}$ satisfying
$\matset{A}\cap\matset{L}_i=\{\mat{0},\mat{A}_i\}$ (continuing the
example, we can take $\matset{A}$ as the set of all matrices $\mat{A}
=(a_{ij})\in \field{2}^{2\times 3}$ satisfying
$a_{12}+a_{13}+a_{23}=a_{12}+a_{22}+a_{23}=0$). The set $\matset{A}$
meets every line of $\RAG(2,3,\field{2})$ in $2$ points and hence
determines a $1$-good random $2\times 3$ matrix $\rmat{A}$ (by
assigning probability $\frac{1}{16}$ to every element of
$\matset{A}$).  But $\matset{A}$ does not contain a $(2,3,1)$ MRD
code, since such a code would intersect some coset of the subspace
$\{\mat{0},\mat{A}_1,\mat{A}_2,\mat{A}_3\}$ of $\matset{A}$ in at
least $2$ points and hence contain $2$ points $\mat{A}$,
$\mat{A}+\mat{A}_i$ of some line $\mat{A}+\matset{L}_i$. Consequently
$\rmat{A}$ cannot be written as a convex combination of MRD codes.

The determination of all vertices of $\grmat_k(m,n,\field{q})$ seems to be a
challenging problem, even for moderately sized values of $m,n,q$.

We end this paper with another ``generic'' research problem.
\begin{problem}
  Investigate the geometries $\RAG(m,n,\field{q})$ from the viewpoint
  of Galois geometry. In particular determine either exactly (for
  small values of $m,n,q$) or approximately (i.e.\ by bounds) the
  minimum size of blocking sets (maximum size of arcs)
  in $\RAG(m,n,\field{q})$.  
\end{problem}
For what has been done in the classical case (i.e.\ $n=1$), see
\cite{hirschfeld-storme98,hirschfeld-storme01} and various chapters of
\cite{nova2011} for an overview.

\section*{Acknowledgements}

The authors are grateful to the reviewers for their careful reading of
the manuscript and for several helpful comments. The work of the
second author was supported by the National Natural
  Science Foundation of China (Grant No.\ 60872063) and the Chinese
  Specialized Research Fund for the Doctoral Program of Higher
  Education (Grant No.\ 200803351027).
 

\begin{thebibliography}{99}

\bibitem{andrews98}
G.~E. Andrews,
\newblock ``The Theory of Partitions,''
\newblock Cambridge University Press, 1998.

\bibitem{ball11}
S.~Ball,
\newblock \emph{The polynomial method in {G}alois geometries},
\newblock in Beule and Storme \cite{nova2011}, Ch.~5, 103--128.

\bibitem{Barg200209}
  \newblock A.~Barg and G.~D. Forney, Jr.,
  \newblock \emph{Random codes: Minimum distances and error exponents},
  \newblock IEEE Trans. Inf. Theory, \textbf{48} (2002), 2568--2573.

\bibitem{nova2011}
J.~D. Beule and L.~Storme (eds.),
\newblock ``Current Research Topics in {G}alois Geometry,''
\newblock Nova Science Publishers, to appear in 2011.

\bibitem{blokhuis-sziklai-szonyi11}
A.~Blokhuis, P.~Sziklai, and T.~Sz{\H o}nyi,
\newblock \emph{Blocking sets in projective spaces},
\newblock in Beule and Storme \cite{nova2011}, Ch.~3, 61--84.

\bibitem{Bose198007}
  \newblock B.~Bose and T.~R.~N. Rao,
  \newblock \emph{Separating and completely separating systems and
    linear codes},
  \newblock IEEE Trans. Comput., \textbf{29} (1980), 665--668.

\bibitem{brouwer-schrijver78}
A.~E. Brouwer and A.~Schrijver,
\newblock \emph{The blocking number of an affine space},
\newblock J. Combin. Theory Ser.~A, \textbf{24} (1978), 251--253.


\bibitem{clerck-maldeghem95}
F.~de~Clerck and H.~van Maldeghem,
\newblock \emph{Some classes of rank two geometries},
\newblock in F.~Buekenhout (ed.), ``Handbook of Incidence
  Geometry---Buildings and Foundations,'' Ch.~10, 433--475, Elsevier
  Science Publ., 1995.

\bibitem{Cohen199411}
  \newblock G.~Cohen and G.~Z\'{e}mor,
  \newblock \emph{Intersecting codes and independent families},
  \newblock IEEE Trans. Inf. Theory, \textbf{40} (1994), 1872--1881.

\bibitem{Cohen200005}
  \newblock G.~Cohen and G.~Z\'{e}mor,
  \newblock \emph{Copyright protection for digital data},
  \newblock IEEE Commun. Lett., \textbf{4} (2000), 158--160.

\bibitem{Csiszar198207}
  \newblock I.~Csisz\'{a}r,
  \newblock \emph{Linear codes for sources and source networks: Error
    exponents, universal coding},
  \newblock IEEE Trans. Inf. Theory, \textbf{28} (1982), 585--592.

\bibitem{Delsarte197800}
  \newblock P.~Delsarte,
  \newblock \emph{Bilinear forms over a finite field, with
    applications to coding theory},
  \newblock J. Combin. Theory Ser.~A, \textbf{25} (1978), 226--241.

\bibitem{dembowski68}
P.~Dembowski,
\newblock ``Finite Geometries,''
\newblock Springer-Verlag, 1968.

\bibitem{evans92}
A.~B. Evans,
\newblock ``Orthomorphism Graphs of Groups,''
\newblock vol. 1535 of Lecture
  Notes in Math., Springer-Verlag, 1992.

\bibitem{Gabidulin198501}
  \newblock E.~M. Gabidulin,
  \newblock \emph{Theory of codes with maximum rank distance},
  \newblock Probl. Inf. Transm., \textbf{21} (1985), 1--12.

\bibitem{Gallager196300}
  \newblock R.~G. Gallager,
  \newblock ``Low-Density Parity-Check Codes,''
  \newblock MIT Press, Cambridge, MA, 1963.

\bibitem{greferath-schmidt99}
M.~Greferath and S.~E. Schmidt,
\newblock \emph{Finite-ring combinatorics and {M}ac{W}illiams'
  equivalence theorem},
\newblock J. Combin. Theory Ser.~A, \textbf{92} (2000), 17--28.

\bibitem{mhall86}
M.~Hall, Jr.,
\newblock ``Combinatorial Theory,''
\newblock 2nd edition, John Wiley \& Sons, 1986.

\bibitem{wt:egal}
W.~Heise and T.~Honold,
\newblock \emph{Homogeneous and egalitarian weights on finite rings},
\newblock in ``Proceedings of the Seventh International Workshop on
  Algebraic and Combinatorial Coding Theory ({ACCT}-2000),''
  Bansko, Bulgaria, 2000, 183--188.

\bibitem{Hirschfeld199800}
  \newblock J.~W.~P. Hirschfeld,
  \newblock ``Projective Geometries over Finite Fields,''
  \newblock 2nd edition, Oxford University Press, New York, 1998.

\bibitem{hirschfeld-storme98}
J.~W.~P. Hirschfeld and L.~Storme,
\newblock \emph{The packing problem in statistics, coding theory and finite
  projective spaces},
\newblock J. Statist. Plann. Inference, \textbf{72} (1998), 355--380.

\bibitem{hirschfeld-storme01}
J.~W.~P. Hirschfeld and L.~Storme,
\newblock \emph{The packing problem in statistics, coding theory and finite
  projective spaces: update 2001},
\newblock in ``Finite geometries,'' vol.~3 of Dev. Math.,
  Kluwer Acad. Publ., Dordrecht, 2001, 201--246.

\bibitem{st:homippi}
T.~Honold and A.~A. Nechaev,
\newblock \emph{Weighted modules and representations of codes},
\newblock Probl. Inf. Transm., \textbf{35} (1999), 205--223.

\bibitem{jamison77}
R.~E. Jamison,
\newblock \emph{Covering finite fields with cosets of subspaces},
\newblock J. Combin. Theory Ser.~A, \textbf{22} (1977), 253--266.

\bibitem{Kleitman197300}
  \newblock D.~J. Kleitman and J.~Spencer,
  \newblock \emph{Families of k-independent sets},
  \newblock Discrete Math., \textbf{6} (1973), 255--262.

\bibitem{Korner199500}
  \newblock J.~K\"{o}rner,
  \newblock \emph{On the extremal combinatorics of the Hamming space},
  \newblock J. Combin. Theory Ser.~A, \textbf{71} (1995), 112--126.

\bibitem{lam91}
T.-Y. Lam.
\newblock ``A First Course in Noncommutative Rings,''
\newblock no. 131 of Grad. Texts in Math.,
Springer-Verlag, 1991.

\bibitem{ivan-storme11}
I.~Landjev and L.~Storme,
\newblock \emph{Galois geometries and coding theory},
\newblock in Beule and Storme \cite{nova2011}, Ch.~8, 185--212.

\bibitem{niederreiter-robinson82}
H.~Niederreiter and K.~H. Robinson,
\newblock \emph{Complete mappings of finite fields},
\newblock J. Aust. Math. Soc. Ser.~A,  \textbf{33} (1982), 197--212.

\bibitem{Pradhan197609}
  \newblock D.~R. Pradhan and S.~M. Peddy,
  \newblock \emph{Techniques to construct $(2, 1)$ separating systems
    from linear error-correcting codes},
  \newblock IEEE Trans. Comput., \textbf{25} (1976), 945--949.

\bibitem{Roth199103}
  \newblock R.~M. Roth,
  \newblock \emph{Maximum-rank array codes and their application to
    crisscross error correction},
  \newblock IEEE Trans. Inf. Theory, \textbf{37} (1991), 328--336.


\bibitem{wan96a}
Z.-X. Wan.
\newblock ``Geometry of Matrices,''
\newblock World Scientific, 1996.

\bibitem{Yang200904}
  \newblock S.~Yang, Y.~Chen, and P.~Qiu,
  \newblock \emph{Linear-codes-based lossless joint source-channel
    coding for multiple-access channels},
  \newblock IEEE Trans. Inf. Theory, \textbf{55} (2009), 1468--1486.

\bibitem{Yang200909}
  \newblock S.~Yang, T.~Honold, Y.~Chen, Z.~Zhang, and P.~Qiu,
  \newblock \emph{Constructing linear codes with good spectra},
  \newblock IEEE Trans. Inf. Theory, submitted for publication,
    \href{http://arxiv.org/pdf/0909.3131}{{\mdseries\ttfamily arXiv:0909.3131}}.

\bibitem{yuan-tong-zhang07}
Y.~Yuan, Y.~Tong, and H.~Zhang,
\newblock \emph{Complete mapping polynomials over finite field {$\field{16}$}},
\newblock in ``Arithmetic of finite fields,'' vol.\ 4547 of Lecture
  Notes in Comput. Sci., Springer-Verlag, Berlin, 2007, 147--158.

\end{thebibliography}

\medskip

Received May 2011; revised July 2011.


\medskip
 {\it E-mail address: }yangst@codlab.net\\
 \indent{\it E-mail address: }honold@zju.edu.cn
\end{document}